\documentclass[twocolumn,aps,prr,letterpaper,superscriptaddress,longbibliography]{revtex4-2}
\usepackage{amsmath,amssymb,amsfonts}
\usepackage{amsthm, blkarray}
\usepackage{mathtools}
\usepackage{hyperref}
\usepackage{float}
\usepackage{graphicx}
\usepackage{enumerate}
\usepackage[dvipsnames]{xcolor}
\usepackage{physics}
\usepackage{xr}
\usepackage[dvipsnames]{xcolor}
\usepackage{tikz}
\usetikzlibrary{quantikz}

\newcommand{\mi}{\mathrm{i}}
\newcommand{\me}{\mathrm{e}}

\newcommand{\cC}{\mathcal{C}}

\newcommand{\cE}{\mathcal{E}}
\newcommand{\cF}{\mathcal{F}}
\newcommand{\cG}{\mathcal{G}}

\newcommand{\cI}{\mathcal{I}}

\newcommand{\cM}{\mathcal{M}}

\newcommand{\cP}{\mathcal{P}}

\newcommand{\cS}{\mathcal{S}}

\newtheorem{theorem}{Theorem}
\newtheorem{lemma}[theorem]{Lemma}

\begin{document}
	
\title{Fault tolerance against amplitude-damping noise using Bacon-Shor codes}
	
\author{Long D. H. My}
\thanks{Both authors contributed equally.}
\affiliation{Yale-NUS College, Singapore}
\affiliation{Centre for Quantum Technologies, National University of Singapore, Singapore}
\author{Akshaya Jayashankar}
\thanks{Both authors contributed equally.}
\affiliation{Centre for Quantum Engineering, Research and Education, TCG CREST, Sector V, Salt Lake, Kolkata, India~700091}
\author{Prabha Mandayam}
\affiliation{Department of Physics, Indian Institute of Technology Madras, Chennai, India~600036}
\affiliation{Center for Quantum Information, Computing and Communication, Indian Institute of Technology Madras, Chennai, India~600036}
\author{Hui Khoon Ng}
\affiliation{Yale-NUS College, Singapore}
\affiliation{Centre for Quantum Technologies, National University of Singapore, Singapore}
\email{huikhoon.ng@nus.edu.sg}

\begin{abstract}
Designing efficient fault tolerance schemes is crucial for building useful quantum computers. Most standard schemes assume no knowledge of the underlying device noise and rely on general-purpose quantum error-correcting (QEC) codes capable of handling arbitrary errors. Biased-noise alternatives focus on only correcting a subset of some generic error basis (e.g., Pauli error basis), and lower resource needs by channeling the redundancy to dealing only with that subset. Yet, the most resource-efficient codes are expected to be those that directly target the specific noise process that afflicts the quantum device, rather than using a generic error-basis description. However, the question of whether such noise-adapted QEC protocols are amenable to fault-tolerant implementations remains largely unexplored. Here, we design a fault tolerance scheme based on the Bacon-Shor codes which can protect against amplitude-damping noise in the device. We construct a universal set of logical gadgets tolerant to multiple damping errors and estimate the fault tolerance threshold of our scheme. Our work thus establishes the possibility of achieving fault tolerance against amplitude-damping noise using noise-adapted quantum codes, while highlighting some of the unique challenges that arise in this context. 
\end{abstract}
	
\maketitle

\section{Introduction}

Quantum fault tolerance \cite{knill96, aharonov96, shor_FT,preskill_FT, gottesman_FT} provides the basic underpinnings of today's global efforts towards realizing robust and scalable quantum computers. The framework of fault tolerance (FT) builds upon the theory of quantum error correction (QEC) \cite{shor, calderbank, knill, gottesman} and offers a pathway to reliable quantum computation in the presence of noisy memory and gates. In a typical FT scheme, one chooses a QEC code to encode the computational---or ``logical"---information into the physical qubits. Encoded operations are performed on the logical qubits, interleaved with periodic error correction to prevent catastrophic accumulation and spread of errors. FT theory offers a suite of techniques to construct noise-resilient logical quantum circuits from noisy physical gates and imperfect error correction operations, provided the physical noise is below a certain threshold level.

The traditional approach to fault tolerance assumes no knowledge---apart from mild assumptions expected to be broadly applicable---of the underlying noise in the device, and relies on general-purpose codes capable of correcting arbitrary errors. This requirement for robustness against arbitrary noise aligns well with the general philosophy of fault tolerance: We ward against imprecision in our knowledge of the noise, by not relying on any such information. Nevertheless, this desire to accommodate arbitrary errors comes at the cost of requiring code choices that can be more resource-demanding. For example, the smallest qubit code capable of correcting arbitrary errors on any one physical qubit demands five physical qubits~\cite{laflamme}; instead, three physical qubits suffice if we assume that there is only bit-flip or only phase-flip noise, or four physical qubits if we assume only amplitude-damping noise \cite{leung}.

Potential resource savings by incorporating properties of the noise are expected to be important for early fault-tolerance experiments when qubit numbers and gate fidelities are limited. Many early-stage error correction demonstrations focus on correcting only say bit-flip or phase-flip noise (see, for example, among many others, Refs.~\cite{schindler2011experimental,Google2021,livingston2022}, because they have only very few qubits or can only do a few gates, too few for general-purpose codes. While some of the recent experiments have taken impressive steps towards fault-tolerant operation \cite{QuEra2024,Google2024}, for scaling to large-scale computers, \emph{noise-adapted} optimizations remain important. General-purpose codes are most efficient when the noise genuinely has no particular asymmetry or special properties, so that one is not ``wasting" the redundancy introduced to protect against noise in all possible ways. Yet, real device noise is never truly {symmetric}, and one could seek noise-adapted fault-tolerance designs to first reduce those noise asymmetries by correcting the dominant noise types, before concatenating---at a higher effective level---with a general-purpose code to protect against residual, now more symmetric, noise.

Previous noise-adapted schemes have largely approached the problem from {the following} view of asymmetry in the noise: We begin with an error basis---one that describes all possible errors---and declare that only a subset of the errors in that basis are dominant, and focus on correcting only those errors. This is the case for the biased Pauli noise models, where one begins with the Pauli error basis and chooses to focus on only $X$ or $Z$ errors. This motivated the many previous discussion on repetition codes~\cite{gourlay, aliferis_biased, aliferis_biasedExp} as well as rectangular surface codes for biased Pauli noise (see, for example, Refs~\cite{biasedNoise_2018, tuckett2019}). Similarly, in bosonic cat codes situations, where errors are decomposed as boson loss or phase errors, one chooses to focus on either loss or phase, depending on the dominant (sometimes engineered) noise process \cite{ofek2016extending,chamberland2022building,guillaud2019repetition,darmawan2021practical,grimm2020stabilization}.

Starting from an error basis and then choosing to drop certain types of errors has a strong advantage when designing FT schemes. One could start from standard general-purpose FT schemes and adopt those designs in the noise-adapted situation, by simply pruning out parts of the FT circuits that deal with the dropped error types. This is particularly obvious for CSS-type codes with biased Pauli noise---one drops all the circuit components that implement the $X$-type checks (which deal with $Z$ errors) when correcting only $X$ errors and not $Z$ errors. This, however, begs the question: rather than pruning existing FT designs, can we start from only the actual noise description---as a quantum process deduced from device characterization---and construct a fully fault-tolerant computational scheme?

We answered this question in part in our earlier work \cite{jayashankar2022achieving}, where we constructed a universal set of fault-tolerant logical operations as well as a fault-tolerant error correction gadget starting from the assumption that the noise is described by the amplitude-damping channel. Amplitude-damping noise cannot be phrased as biased noise in terms of the Pauli error basis, and is the prototypical example that many studies turn to when hoping to examine beyond-Pauli situations. There are thus well-studied codes available that focus on correcting amplitude-damping noise~\cite{leung, langshor, fletcherpaper, cartanform_2020}. Fault-tolerance studies of such codes were lacking, and our earlier work in \cite{jayashankar2022achieving} offered a first full construction of fault-tolerant gadgets for error correction and universal computation. There, we also pointed out the differences in circuit designs when discussing fault-tolerance for such channel-adapted situations, compared to standard Pauli or biased-Pauli scenarios. For example, that transversal constructions are always fault tolerant in Pauli settings, no longer hold true in channel-adapted situations.

In our present work, we extend and complete our answer for amplitude-damping noise by showing that we can scale up our fault-tolerant designs to larger codes. Our earlier work \cite{jayashankar2022achieving} examined only the case of the smallest code, which corrects a single amplitude-damping error in any of the physical qubits within a single gadget. Fault-tolerant schemes come also with a prescription of how to scale up the protection, by using more physical qubits and operations, to correct more errors and hence compute more accurately. While we expect, in real-use cases, to eventually move to general-purpose codes to continue scaling up, the asymmetry in the noise may be strong enough to require first increasing protection against the amplitude-damping noise. We address this aspect here, and develop a general FT scheme for amplitude-damping noise based on the two dimensional Bacon-Shor code~\cite{bacon_2006}, defined on a square $(t+1)\times (t+1)$ lattice of physical qubits and capable of correcting amplitude-damping erorrs on up to $t$ qubits. The smallest instance, of the $2\times 2$ lattice ($t=1$), corresponds to the 4-qubit code used as the basis of our earlier work \cite{jayashankar2022achieving}. Here, we explain how to adapt our earlier designs for $t>1$, and construct a fault-tolerant error-correction gadget as well as a full set of universal logical operations. We demonstrate the scaling performance by estimating the logical infidelity of quantum memory and showing that the infidelity falls---below specific threshold levels---as the code size grows.
	
The rest of the paper is organized as follows. We start, in Sec.~\ref{Sec:prelims}, by reviewing the amplituding-damping noise as well as the $(t+1)\times (t+1)$ Bacon-Shor code.  In Sec.~\ref{sec:ecunit}, we describe our construction of a fault-tolerant error correction gadget for the Bacon-Shor code for general $t$. In Sec.~\ref{sec:czgadget}, we focus on our construction of the logical \textsc{CZ} gadget, leaving the details of other logical gadgets to the appendix (Apps.~\ref{sec:ftgadgets} and \ref{sec:universal}) as they generalize easily from our earlier work. Finally, we estimate the logical fidelity and pseudothreshold for the quantum memory task using our scheme in Sec~\ref{sec:threshold}, and conclude in Sec.~\ref{sec:conclude}.

\section{Preliminaries}\label{Sec:prelims}
We begin by reviewing the assumed noise model, the Bacon-Shor codes, and the ideal (no-noise) error-correction procedure, in preparation for the description of our fault-tolerant construction. 
	
\subsection{Noise model}\label{sec:noisemodel}
	
We assume amplitude damping as the dominant noise affecting every physical qubit in our quantum computer. Physical processes like spontaneous decay give rise to amplitude-damping noise and this is the leading source of noise in some quantum computing platforms (see, for example, Ref.~\cite{scqubit}). Amplitude-damping noise is a single-qubit channel described by the completely positive (CP) and trace-preserving (TP) map,
\begin{align}
		\label{eq:ampdamp}
		\cE_\mathrm{AD}(\,\cdot\,)&=E_0(\,\cdot\,)E_0^\dagger + E_1(\,\cdot\,)E_1^\dagger,\\
		\textrm{with }		
		E_0&\equiv\frac{1}{2}{\left[(1+\!\sqrt{1-p})I + (1-\!\sqrt{1-p}) Z\right]}\nonumber\\
		&=\ketbra{0}{0}+\sqrt{1-p}\ketbra{1}{1}, \nonumber \\
		\textrm{and } E_1 &\equiv \frac{1}{2}\sqrt{p}(X+ \mi Y)=\sqrt p \ketbra{0}{1} \equiv \sqrt{p}E .\nonumber
\end{align}
Here $I$ is the single-qubit identity, $X, Y$, and $Z$ are the Pauli operators, and $\{|0\rangle, |1\rangle\}$ is the $Z$ eigenbasis. $p\in[0,1]$ is the damping parameter, assumed to be small for the weak noise regime needed for quantum computing. Below, we write $E\equiv \ketbra{0}{1}$ for the damping operator.
	
Similar to the the FT scheme of Ref.~\cite{jayashankar2022achieving}, we use the following unencoded (physical) operations to build the FT encoded gadgets:
\begin{equation} \label{eq:fault-tol}
		\{ \cP_{|+\rangle}, \cP_{|0\rangle}\} \cup \{ \cM_{X},\cM_{Z},\textsc{CNOT},\textsc{CZ},X, Z, S, T\} .
\end{equation}
Here,  $\cP_{|+\rangle}$ and $\cP_{|0\rangle}$  refer to the preparation of the eigenstates of $X$ and $Z$, respectively, and $\cM_{X}$ and  $\cM_{Z}$ are measurements in the $X$ and $Z$ bases, respectively. \textsc{CNOT} $\equiv\ketbra{0}{0}\otimes I+\ketbra{1}{1}\otimes X$ is the two-qubit controlled-\textsc{NOT} gate, \textsc{CZ} $\equiv\ketbra{0}{0}\otimes I+\ketbra{1}{1}\otimes Z$ is the two-qubit controlled-$Z$ gate, $S\equiv \ketbra{0}{0}+\mi\ketbra{1}{1}$ is the phase gate, and $T\equiv\ketbra{0}{0}+\me^{-\mi\pi/4}\ketbra{1}{1}$ is the $\pi/8$ gate.
	
We assume that errors in the state preparations, gates, and measurements are all due to $\cE_\mathrm{AD}$. Specifically, a noisy physical gate $\cG$ is modeled as the ideal gate followed by the noise $\cE_{AD}$ acting independently on each qubit participating in the gate. A noisy measurement is modeled as an ideal measurement preceded by the noise $\cE_\mathrm{AD}$, while a noisy preparation is an ideal preparation followed by $\cE_\mathrm{AD}$. We emphasize that the noise acts on each physical qubit individually, and is assumed to be time and gate independent. More generally, one can regard the parameter $p$ as an upper bound on the level of amplitude damping over time and gate variations.
	
In our discussion below, it is useful to organize the effect of $\cE_{AD}$ in powers of $p$. To that end, we write
\begin{align}\label{eq:ampdamp2}
		\cE_\mathrm{AD}(\,\cdot\,)&=\tfrac{1}{4}{\left(1+\sqrt{1-p}\right)}^2\cI(\,\cdot\,)+p\cF_a(\,\cdot\,)\\
		&\quad + \tfrac{p}{2}\cF_z(\,\cdot\,)+{\left[\tfrac{p^2}{16}+O(p^3)\right]}Z(\,\cdot\,)Z,\nonumber
\end{align}
where $\cI(\cdot)\equiv (\cdot)$ is the identity channel, $\cF_a(\cdot)\equiv E(\cdot)E^{\dagger}$ is the \emph{damping error}, and $\cF_z(\cdot)\equiv \tfrac{1}{2}[(\cdot)Z + Z(\cdot)]$ is referred to as the \emph{off-diagonal error}. The latter is sometimes referred to as the backaction error in the context of superconducting qubits~\cite{liang}. Written in this manner, $\cE_{AD}$ can be thought of as leading to no error (or order 1) when the $\cI$ part occurs, a first-order (in $p$) error when the $\cF_a$ or $\cF_z$ part occurs, and a second-order error when a $Z$ error occurs.

\subsection{Bacon-Shor codes}\label{sec:baconshor}

Bacon-Shor codes are a class of stabilizer codes originally proposed in the context of \emph{subsystem} quantum error correction~\cite{bacon2006operator, kribs2005operator}. Subsequently, it was shown that the class of Bacon-Shor codes achieves a higher fault-tolerance threshold than subspace stabilizer codes while lending itself to simpler fault-tolerant implementations that require only nearest-neighbour two-qubit measurements \cite{aliferis_baconshor}. The $d \times d$ Bacon-Shor code is a distance-$d$ CSS code encoding one logical qubit in $d^{2}$ physical qubits. For small lattice sizes, the Bacon-Shor code was shown to yield an FT threshold of $\sim 10^{-4}$ against stochastic adversarial noise~\cite{aliferis_baconshor}. More recently, the $[[4,1,2]]$ Bacon-Shor code was shown to achieve a pseudothreshold between $3$-$6\%$ in the presence of depolarizing noise~\cite{shlosberg2021}. However, increasing the size of the code does not improve the logical rate indefinitely; rather, it reverses beyond a point. The code thus does not have a quantum accuracy threshold (i.e., the physical noise level below which increasing the code size reduces the logical error without limits), a problem which can be addressed to some extent by concatenation~\cite{aliferis_baconshor} or by lattice-surgery techniques~\cite{gidney2023}. We will see this absence of a true FT threshold also in our construction, but this reversal point occurs at a large code distance, likely beyond the regime we expect to use such noise-adapted codes (see Sec.~\ref{sec:threshold}).

It was first noted in Ref.~\cite{duan2010multi} that the $(t+1)\times(t+1)$ Bacon-Shor code can correct $t$ amplitude-damping errors. Note that a code is said to correct $t$ amplitude damping errors if it is able to reduce the infidelity of the encoded state, to order $O(p^{t+1})$ after the error correction. Subsequently, explicit error correction schemes were constructed for the Bacon-Shor code in Ref.~\cite{piedrafita2017reliable}.

\begin{figure}
		\includegraphics[scale=0.5]{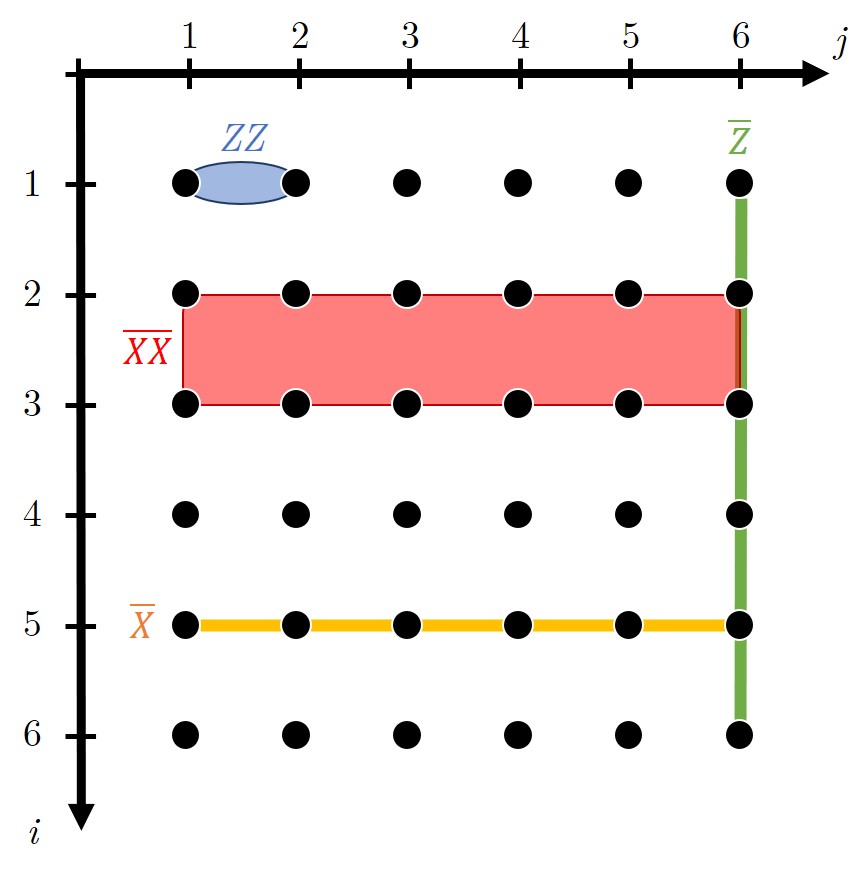}
		\caption{\label{fig:BSCode}A $6 \times 6$ Bacon-Shor code in the $Z$ gauge. A $Z$-type stabilizer (marked $ZZ$ in blue), an $X$-type stablizer (marked $\overline{X}\overline{X}$ in red), the logical $X$ operator (marked $\overline{X}$ in orange), and the logical $Z$ operator (marked $\overline{Z}$ in green) are shown.}
\end{figure}
	
Here, we follow the proposal of Refs.~\cite{duan2010multi,piedrafita2017reliable}, which employs the Bacon-Shor code with the choice of the $Z$ gauge; see Fig.~\ref{fig:BSCode}. The two-dimensional $n\times n$ Bacon-Shor code in the $Z$ gauge is a subspace stabilizer code that encodes one qubit into $n^2$ qubits arranged in an $n\times n$ square lattice. The stabilizer group of the code is generated by the operators $Z_{i,j}Z_{i,j+1}$ [$1\leq i \leq n, 1 \leq j \leq n-1$, with $i(j)$ as the row(column) index], acting on neighboring qubits in the same row, and $X_{i,\star}X_{i+1,\star} (1 \leq i \leq n-1)$, acting on all the qubits in two neighboring rows ($\star$ denotes all values of that index slot). This code can be thought of as the concatenation of a length-$n$ repetition code (within each row) that protects against $X$ errors, with a length-$n$ repetition code protecting against $Z$ errors (across the rows). The code space $\cC$ is thus the span of the two states,
\begin{align}\label{eq:codewords}
		\ket{\pm}_L&\equiv {\left[\tfrac{1}{\sqrt{2}}(\ket{0}_\mathrm{row}\pm\ket{1}_\mathrm{row})\right]}^{\otimes n},
\end{align}
where $\ket{0}_\mathrm{row}\equiv \ket{0}^{\otimes n}$ and $\ket{1}_\mathrm{row}\equiv \ket{1}^{\otimes n}$ are the length-$n$ repetition code codewords for the qubits in a single row; the outer tensor product in Eq.~\eqref{eq:codewords} goes over the rows.
The logical operators are identified as $\overline{X}\equiv X_{i,\star}$ and $\overline{Z}\equiv Z_{\star,j}$, denoting, respectively, $X$ operators acting on all the qubits in any one row of the lattice, and $Z$ operators acting on all the qubits in any one column of the lattice. This gives us also a useful shorthand: We write the $X_{i,\star}X_{i+1,\star}$ stabilizers simply as $\overline{XX}$, putting in an $i$ index only when referring to a particular pair of rows.
	
The fact that the $n\times n$ Bacon-Shor code can correct amplitude-damping errors on $t\equiv n-1$ qubits can be seen by checking that the code and the amplitude-damping error set satisfy the Knill-Laflamme error-correction conditions \cite{knill1997theory} up to order $O(p^t)$ \cite{gottesman97, duan2010multi}. In fact, the well-known $4$-qubit code \cite{leung} that corrects a single amplitude-damping error is a $2\times 2$ Bacon-Shor code. We note in contrast that the $n\times n$ Bacon-Shor code can correct \emph{arbitrary} errors on only $\sim\!n/2$ qubits.

	
\subsection{Ideal EC procedure}\label{sec:idealEC}
	
An EC procedure to correct amplitude-damping errors with Bacon-Shor codes was constructed in Ref.~\cite{piedrafita2017reliable}, employing only Pauli measurements and gates. This EC procedure is not fault tolerant and functions correctly only when there are no faults in the EC circuit itself. We thus refer to this as the ``ideal" EC procedure. We note that this ideal EC procedure may be suboptimal (see, for example, Ref.~\cite{fletcherpaper} for a numerically optimized recovery) in that it reduces the infidelity of the output state to the requisite $O(p^{t+1})$ but it need not give the minimal infidelity among all recovery maps. Its simple structure, however, provides the starting point for our fault-tolerant version below. We thus briefly review the components of the ideal EC scheme here; more details are given in App.~\ref{app:err_corr}. 
	
The ideal EC procedure involves three steps to deal with the $\cF_a$, $\cF_z$, and $Z$ errors (see Eq.~\ref{eq:ampdamp2}) that can occur from amplitude-damping noise:
\begin{enumerate}[(1)]
		\item \textit{Damping syndrome extraction}. We measure all $Z_{i,j}Z_{i,j+1}$ stabilizers to deduce which rows of qubits suffered damping errors $\cF_a$. This is followed by individual $Z$ measurements on all qubits in the damped rows to further locate the errors. We then apply an $X$ to each damped qubit to remove the $X$ part of the damping error $E=X(I-Z)$, leaving the $(I-Z)$ part for the next step. 
		\item \textit{$\overline{XX}$ measurements}. Next, we measure all $\overline{XX}$ stabilizers, to simultaneously perform damping-error recovery and deal with $Z$-type errors (that is, $Z$, $\cF_Z$, or $I\pm Z$ errors). On damped rows (those which had $\cF_a$ errors), as discovered by the earlier step, the $\overline{XX}$ measurement projects the state back to the correct code state, up to $Z$ errors that can be deduced from the measurement outcomes. On all rows, the $\overline{XX}$ measurement has the effect of either annihilating the $\cF_z$ errors, or detecting $Z$ errors. 
		\item \textit{Recovery for $Z$ errors}. The $Z$ errors can be deduced and remedied in the recovery step by combining the syndrome information from all the $ZZ$ and $\overline{XX}$ measurements of the earlier two steps.
\end{enumerate}
	
The ideal EC procedure is clearly not fault tolerant. For example, a single $\cF_a$ in the data qubits right before the $\overline{XX}$ measurement is converted by that measurement into an $X$ error. This $X$ error is uncorrectable for the $(t+1)\times (t+1)$ Bacon-Shor code tailored for amplitude-damping noise as it can be confused with $t$ damping errors. In the following section, we describe a construction of the EC gadget that \emph{is} fault tolerant.

\section{Fault-tolerant EC gadget}
\label{sec:ecunit}
	
\begin{figure*}
\includegraphics[width=\textwidth]{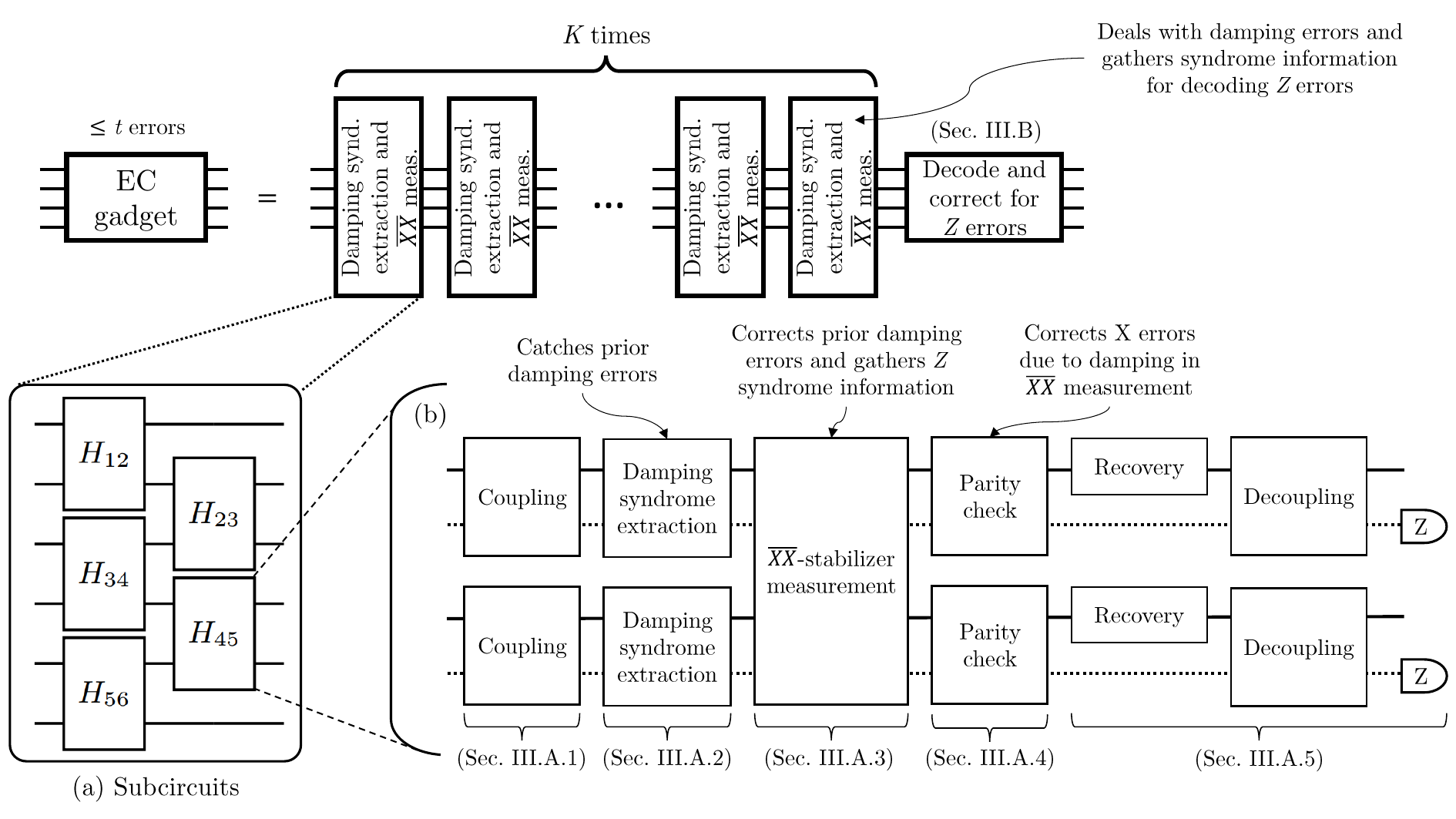}
\caption{\label{fig:ECBigPicture} Schematic overview of the fault-tolerant EC gadget, with references to the sections in the main text where the individual components are explained. Each solid horizontal line represents a row of $t+1$ data qubits; each dotted line represents a row of $t$ ancillary qubits.}
\end{figure*}
	
Our goal here is to construct an EC gadget that satisfies standard FT requirements (see App.~\ref{app:ftprops} for a statement of these requirements). A schematic overview of the whole fault-tolerant EC gadget is given in Fig.~\ref{fig:ECBigPicture}. We first note that the ideal EC procedure described above is composed of smaller subcircuits $H_{i,i+1}$ [see Fig.~\ref{fig:ECBigPicture}(a)] since the stabilizer measurements involve at most two neighboring rows. Our first step towards constructing a fault-tolerant \textsc{EC} gadget is thus to make the subcircuits fault tolerant against the first-order damping error $\cF_a$, while containing the spread of $Z$ errors. We describe this construction in Sec.~\ref{sec:subcircuit}. Subsequently, in Sec.~\ref{sec:decodingz}, we discuss how to obtain and decode the syndrome for $Z$ errors in a fault-tolerant manner using these subcircuits, thus completing our construction of a fault-tolerant  $\textsc{EC}$ gadget as a combination of individual subcircuits fault tolerant against $\cF_a, \cF_z$, and $Z$ errors.

\subsection{Fault-tolerant subcircuits}\label{sec:subcircuit}

A subcircuit is fault tolerant against $\cF_a$ if it satisfies the following condition: For an incoming state to the subcircuit with $s$ $\cF_a$ errors, if $r$ $\cF_a$ errors occur during the subcircuit such that $s+r\leq t$, the output of the subcircuit differs from the ideal output by at most $r$ $\cF_a$ errors. This property guarantees that the whole EC  circuit constructed from many subcircuits together is also fault tolerant against $\cF_a$.
	
To achieve this, we exploit a similar technique used in Ref.~\cite{jayashankar2022achieving} to make the $4$-qubit code fault-tolerant against a single damping error: We temporarily extend each row to a $(2t+1)$-repetition code, with the use of $t$ ancillary qubits, which protects against $t$ $X$ errors without confusion with $t$ damping errors. We switch back to having only $t+1$ qubits in each row at the end of the subcircuit. Each subcircuit $H_{i,i+1}$ is made up of several components as shown in Fig.~\ref{fig:ECBigPicture}(b). We describe each component separately below.  

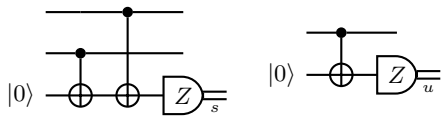
\begin{figure}
	\begin{quantikz}[row sep=0.05cm, column sep = 0.3cm]
		\ghost{X} & \lstick{} & \qw & \ctrl{2} & \qw \\
		\ghost{X} & \lstick{} & \ctrl{1} & \qw & \qw \\
		\ghost{X} & \lstick{$\ket{0}$} & \targ{} & \targ{} & \meterD{Z} & \cw{s}
	\end{quantikz}\quad
	\begin{quantikz}[row sep=0.05cm, column sep = 0.3cm]
		\ghost{X} & \lstick{} & \ctrl{1} & \qw \\
		\ghost{X} & \lstick{$\ket{0}$} & \targ{} & \meterD{Z} & \cw{u}
	\end{quantikz}
	\caption{Circuits for (a) a parity measurement between two qubits (left circuit), and (b) the non-destructive measurement of a qubit in the $Z$ basis (right circuit).}
	\label{fig:paritymeas}
\end{figure}

\subsubsection{Coupling} \label{sec:subC_coupling}
	
In the first step, we extend each row from $t+1$ to $2t+1$ qubits by coupling $t$ out of the $t+1$ data qubits to a set of $t$ ancillary qubits initialized to the $\ket{0}$ state. Every data qubit in the set is coupled to a single ancillary qubit via a \textsc{CNOT}, with the data qubit as the control and the ancilla as the target. Each row thus temporarily becomes a $(2t+1)$-repetition code capable of protecting against $t$ $X$ errors. We refer to these ancillary qubits coupled to the data qubits as the \emph{coupling} ancillas, to distinguish them from ancillary qubits used elsewhere. The $(2t+1)$ data and coupling ancillas are referred to as \emph{extended} data qubits.

\subsubsection{Damping syndrome extraction}\label{sec:damping}
	
Next, we have a damping syndrome extraction step to detect $\cF_a$ errors occurring before and after the coupling step. Note that, in the coupling step, the \textsc{CNOT}s can propagate damping errors from the data qubits to the coupling ancillas, turning $s$ damping errors in the incoming state to $2s$ damping errors in the set of extended data qubits. If these errors are not taken care of before the $\overline{XX}$ measurement, they get converted into $2s$ $X$ errors that, together with further faults along the way, may exceed the capacity of the $(2t+1)$-repetition code. A damping detection step is thus necessary immediately after the coupling step to catch and fix these errors, in addition to any incoming damping errors.
	
The damping detection is performed for each row separately. We start by measuring $Z_{i,j}Z_{i,j+1} (1\leq j \leq t)$ stabilizers between two neighboring data qubits in the same row, using the parity measurement circuit of Fig.~\ref{fig:paritymeas}(a), employing an additional ancilla initialized to $\ket{0}$. At this stage, we only need to measure the parity of the data qubits, and not the coupling ancillas. Measuring the parity of the data qubits detects any of the $s \leq t$ damping errors occurring at the input or after the \textsc{CNOT} gates. The damping errors arising independently on the coupling ancillas (at most $t$) can be identified and corrected in the subsequent units once the errors on the data qubits have been fixed. Parity checks on the data qubits only thus suffice here.

One round of damping syndrome extraction applied to one row of data qubits involves $t$ parity measurements between adjacent pairs of qubits, giving a length-$t$ syndrome sequence of $\pm 1$. A nontrivial sequence indicates that damping errors have occurred in the data qubits and/or during the measurement. The state of the extended data qubits in that row is then reduced to a $Z$ eigenstate (i.e., a string of $0$s and $1$s). At this point, we measure both the data qubits and the coupling ancillas in the $Z$ basis using the circuit of Fig.~\ref{fig:paritymeas}(b); this measurement tells us which qubits have been damped to $\ket{0}$. Finally, we apply $X$ to the damped qubits to remove the $X$ part of the damping error $E=X(I-Z)$. This ensures we do not carry forward more $X$ errors than can be handled by the repetition code, as explained below.
	
If a trivial all $+1$ syndrome sequence is obtained after the first round of damping syndrome extraction, we cannot immediately conclude that there is no damping error---additional measurement faults could have given a falsely trivial syndrome. Instead, we have to repeat the parity measurements and keep doing so if we continue to get a trivial syndrome, until we have done $t$ rounds. At this point, we can conclude that the row is undamped: Since at least one fault is needed for a nontrivial sequence, either the row is genuinely undamped, or more than $t$ faults have occurred so far and this is beyond the scope of an FT syndrome measurement (see property P3 in App.~\ref{app:ftprops}).
	
If a nontrivial sequence is obtained at any point during the repeated rounds of damping syndrome extraction, we stop the parity measurements and non-destructively measure all the extended data qubits in the faulty row in the $Z$ basis to detect which qubits have been damped. 
	
We note that, for our noise model, if a qubit has been damped to $\ket{0}$, a $Z$ measurement---ideal or faulty---always gives the $+1$ outcome, and we never miss a damping event. However, a faulty $Z$ measurement could give $+1$ outcome for a qubit that started in the state $\ket{1}$ and we might wrongly apply an $X$ to an undamped qubit, effectively turning it into a damped qubit. We could have one of two possibilities happen here: (i) At least one parity measurement has a non-trivial outcome and the entire row of extended data qubits is measured; in this case, the number of errors in the outgoing state (equivalently, the number of errors in the state entering the next step) is limited by the number of faulty $Z$ measurements. (ii) No parity measurement has a non-trivial outcome; in this case, the number of errors in the outgoing state is limited by the number of errors on the coupling ancillas during the parity measurements and on the data qubits during the last layer of the parity measurements. In both cases, we note that the number of errors in the outgoing state is no more than the number of errors occurring inside the syndrome extraction step.

\subsubsection{$\overline{XX}$ measurement} \label{sec:XXmeas}

\begin{figure}
\begin{quantikz}[row sep=0.05cm, column sep = 0.2cm]
	& \lstick{1} & \targ{} & \qw & \qw & \qw & \qw & \qw & \qw & \qw & \qw & \qw & \qw \\
	& \lstick{2} & \qw & \qw & \targ{} & \qw & \qw & \qw & \qw & \qw & \qw & \qw & \qw \\
	& \lstick{3} & \qw & \qw & \qw & \qw & \targ{} & \qw & \qw & \qw & \qw & \qw & \qw \\
	& \lstick{4} & \qw & \qw & \qw & \qw & \qw & \qw & \targ{} & \qw & \qw & \qw & \qw \\
	& \lstick{5} & \qw & \qw & \qw & \qw & \qw & \qw & \qw & \qw & \targ{} & \qw & \qw \\
	& \lstick{6} & \qw & \targ & \qw & \qw & \qw & \qw & \qw & \qw & \qw & \qw & \qw & \qw\\
	& \lstick{7} & \qw & \qw & \qw & \targ{} & \qw & \qw & \qw & \qw & \qw & \qw & \qw \\
	& \lstick{8} & \qw & \qw & \qw & \qw & \qw & \targ{} & \qw & \qw & \qw & \qw & \qw \\
	& \lstick{9} & \qw & \qw & \qw & \qw & \qw & \qw & \qw & \targ{} & \qw & \qw & \qw \\
	& \lstick{10} & \qw & \qw & \qw & \qw & \qw & \qw & \qw & \qw & \qw & \targ{} & \qw \\
	\lstick{$\ket{+}$} & \ctrl{1} & \ctrl{-10} & \ctrl{-5} & \ctrl{-9} & \ctrl{-4} & \ctrl{-8} &\ctrl{-3} & \ctrl{-7} & \ctrl{-2} & \ctrl{-6} & \ctrl{-1} & \ctrl{1} & \meterD{X} \\
	\lstick{$\ket{0}$} & \targ{} & \qw & \qw & \qw & \qw & \qw & \qw & \qw & \qw & \qw & \qw & \targ{} & \meterD{Z}
\end{quantikz}

\includegraphics[scale=0.6]{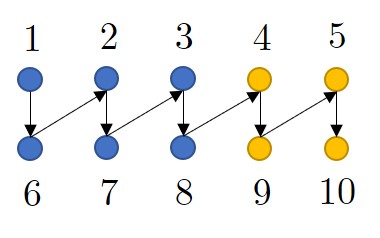}
\caption{The $\overline{XX}$-stabilizer measurement between a pair of successive rows, shown for the $3\times 3$ Bacon-Shor code. The top panel shows the circuit involving the $2\times 5=10$ extended data qubits that make up two successive rows of a $3\times 3$ Bacon-Shor code, with one ancillary qubit and one flag qubit (bottom qubits). The bottom panel shows the order in which the extended data qubits are coupled via the \textsc{CNOT}s to the ancillary qubit to assure fault-tolerance properties (see Sec.~\ref{sec:XXmeas}). The blue circles are the data qubits; the yellow ones are the coupling ancillas; the arrows indicate the order in which the ancillary qubit is coupled to the extended data qubits in the two rows.}
\label{fig:XX_stabilizer}
\end{figure}

We now move to the measurement of the $\overline{XX}$ stabilizer. As in the ideal EC procedure, this step serves to recover code states from damping errors and deal with $Z$-type errors. A single ancillary qubit is initialized in $\ket{+}$ and then coupled, via \textsc{CNOT}s, to each of the extended data qubits in two successive rows; see Fig.~\ref{fig:XX_stabilizer} for the example of the $3 \times 3$ Bacon-Shor code. The ancillary qubit is then measured in the $X$ basis. In the absence of faults in the $\overline{XX}$ measurement, with errors arising only from incoming errors prior to this $\overline{XX}$ step, an outcome of $\pm1$ indicates the presence or absence of $Z$-errors in the outgoing state. To make this measurement fault-tolerant, we add a flag qubit \cite{chao2018quantum} initialized to $\ket{0}$, as seen in Fig.~\ref{fig:XX_stabilizer}. This qubit, along with a pair of CNOTs and a $Z$ measurement, indicates whether a damping error has occurred in the ancillary qubit during the $\overline{XX}$ measurement.  In the absence of faults, the measurement of the flag qubit always yields the outcome $+1$.
	
In the presence of faults in the $\overline{XX}$ measurement, however, the state can be contaminated by additional types of errors. We consider only $\cF_a$ errors during the measurement procedure. $\cF_Z$ and $Z$ errors on the extended data qubits, from faulty \textsc{CNOT}s in the measurement circuit or in the waiting locations after the \textsc{CNOT}s, do not spread to the ancillary qubit and can be treated as incoming errors for the next measurement. The only errors that can cause problems are $\cF_Z$ or $Z$ errors on the ancillary and flag qubits, which can flip the $\overline{XX}$ measurement outcome. We explain how to deal with those errors in Sec.~\ref{sec:decodingz}; here, we focus only on the $\cF_a$ errors. 
	
There are various error scenarios that may occur during the $\overline{XX}$ measurement: 
	\begin{itemize}
		\item A damping error $\cF_a$ on one of the extended data qubits, after the damping syndrome extraction step but before the $\textsc{CNOT}$ on that qubit, is projected to an $X$ ($Y$) error in the outgoing state when the measurement outcome is $+1$($-1$).
		\item A damping error $\cF_a$ on one of the extended data qubits, due to the \textsc{CNOT} in the measurement circuit or from the idling time after the \textsc{CNOT}, causes a damping error in the outgoing state; these are in fact benign for the current measurement and are  detected in the subsequent parity checks.
		\item A damping error $\cF_a$ anywhere on the ancillary qubit propagates an $X$ error to the flag qubit and correlated $X$ errors to the extended data qubits in the outgoing state. Without the flag qubit, this can lead to logical errors when decoding, as it can be confused with the scenario that the extended data qubits are truly damped.

        \item A damping error $\cF_a$ on the flag qubit can propagate a $Z$-type error to the ancillary qubit, and/or flip the flag qubit, but it does not introduce any error to the extended data qubits. Two damping errors on both the flag and the ancillary qubits can lead to a random outcome in the measurement of the ancillary qubit, while flipping the outcome of the flag qubit to $+1$. This can cause a $Z$ error when decoding, but this does not pose a problem since it is an $O(p^2)$ event that leads to an $O(p^2)$ error.
	
\end{itemize}

In what follows, we explain how to perform the $\overline{XX}$ measurement in a fault-tolerant manner in the presence of the error scenarios above. For that, we need the concept of an \emph{error sequence} $e$ for each row of extended data qubits. $e$ specifies the error on each qubit in the outgoing state after the $\overline{XX}$ measurement: We write $X$ if the qubit is affected by an $X$-type error (that is, an $\cF_a$, $X$, or $Y$ error); we write $I$ if there are no errors or only $Z$-type errors.
For example, the error sequences on a pair of successive rows may look like $e_1\equiv X I I X X$ for row $1$ and $e_2\equiv I X I X I$ for row $2$.
For each pair of successive rows with error sequences $e_i$ and $e_{i+1}$, we let $\textsc{diff}(e_i,e_{i+1})$ be the number of columns in which $e_i$ and $e_{i+1}$ differ. For the example error sequences $e_1$ and $e_2$ above, $\textsc{diff}(e_1,e_2)=3$ as the two sequences differ in columns $1$, $2$, and $5$.
	
As we explain below, as long as $\textsc{diff}\leq t$ for every pair of successive rows, we can ensure fault-tolerant operation of the $\overline{XX}$ measurement. That $\textsc{diff}\leq t$ can be assured by a clever ordering of the \textsc{CNOT}s between the extended data qubits and the ancillary qubit, as shown in Fig.~\ref{fig:XX_stabilizer}. This statement is summarized in Lemma \ref{lemma:XXmeasurement} below; the proof is in App.~\ref{app:ProofThm1}.

\begin{lemma}\label{lemma:XXmeasurement}
Let $e_i$ and $e_{i+1}$ be the $X$-type error sequences for a pair of successive rows, caused by $\ell \leq t$ errors in the $\overline{XX}$-stabilizer measurement. It is possible to ensure $\textsc{diff}(e_i, e_{i+1}) \leq \ell$ by fixing the order in which the ancillary qubit is coupled to the extended data qubits, as given in Fig.~\ref{fig:XX_stabilizer}.
\end{lemma}
	
Two points are worth noting about Lemma~\ref{lemma:XXmeasurement}. First, the errors occuring inside the measurement can be any combination of the error scenarios mentioned above. Second, Lemma~\ref{lemma:XXmeasurement} is only applicable if the errors on the two rows in question originate from the $\overline{XX}$ measurement in that subcircuit, and have not been propagated from other rows. Due to the way we execute subcircuits in an EC gadget [see Fig.~\ref{fig:ECBigPicture}(a) and Sec.~\ref{sec:decodingz} below], this condition is satisfied.
	
After performing the $\overline{XX}$ measurement, we treat each row as a $(2t+1)$ repetition code and use the parity measurements of Fig.~\ref{fig:paritymeas}(a) to obtain the syndrome sequence for each row. The syndrome sequence for a given row can, however, correspond to two complementary error sequences $e$ and $\overline{e}$, where $\overline{e}$ is just $e$ with the swap $I\leftrightarrow X$ for each qubit. For example, both $e=XXIII$ and $\overline{e}=IIXXX$ correspond to the same syndrome sequence $(+1,-1,+1,+1)$. 
Lemma~\ref{lemma:XXmeasurement} then serves as our guide to deciding between the two possibilities: Based on the syndrome sequences for each pair of successive rows, we choose the error sequences $e_1$ and $e_2$ associated with the two rows such that $\textsc{diff}(e_1, e_2) \leq t$.

Assuming for the moment that the parity measurements following the $\overline{XX}$ measurement are perfect (imperfect parity measurements are handled in the next section), let us examine how that $\overline{XX}$ measurement behaves under the error scenarios mentioned above; see App.~\ref{app:ProofThm1} for further details.
\begin{itemize}
	\item When the ancilla is faulty, it propagates $X$ errors to the data qubits such that the two error sequences satisfy $\textsc{diff}(e_1, e_2) =1$, due to the way the CNOTs are applied (see Fig.~\ref{fig:XX_stabilizer}).
	\item Any $s\leq t$ incoming damping errors in the extended data qubits propagate as $s\leq t$ $X$ errors. The valid syndrome sequence for the two rows, from the parity check units, is therefore the one that satisfies $\textsc{diff}(e_1, e_2) \leq t$.
	\item Similarly, any $s\leq t$ outgoing damping errors after the $\overline{XX}$ measurement is identified by choosing the syndrome sequence that satisfies $\textsc{diff}(e_1, e_2) \leq t$. 
\end{itemize}
Note that the last two statements above are straightforward to see since both the rows that have $s \leq t$ damping or $X$ errors are correctable by applying $X$ operators on $s \leq t$ qubits on both the rows, such that they naturally satisfy $\textsc{diff}(e_1, e_2) \leq t$. Finally, the case where one could have a combination of the first and one of the other two error scenarios is treated in App.~\ref{app:ProofThm1}. We have thus demonstrated that $t$ damping errors in the $\overline{XX}$ measurement will not lead to uncorrectable $X$ errors on more than $t$ damping errors in the outgoing state.

\subsubsection{Parity checks}
\label{sec:parity_checks}
We also need to deal with potential faults in the parity measurements [Fig.~\ref{fig:paritymeas}(a)] following the $\overline{XX}$ measurement. Faults in the parity measurements can cause two problems. First, they can introduce more errors to the extended data qubits. However, regardless of how those errors are introduced, as long as the total number of errors in the $\overline{XX}$ measurement and the parity-check steps is no more than $t$, the condition of Lemma~\ref{lemma:XXmeasurement} is still satisfied and the error sequences at the end of the parity check step still have the property stated in the lemma. 
		
Second, if new faults occur on the ancilla qubits that are used in the parity measurements, they can give incorrect syndromes and consequently erroneous inferences of error sequences. Again, we address this simply by repeating the parity measurements until we get the same syndrome sequence $t+1$ times. In the presence of no more than $t$ faults, there must be at least one round without any faults among the $t+1$ rounds in which we obtain the same syndrome sequence. By using the syndrome sequence that repeats $t+1$ times, we can confidently infer the correct error sequence. In the worst-case scenario, we may need to repeat up to $t(t+1)+1$ rounds. This takes into account the possibility of fresh errors occurring as we repeat, but we are limited to no more than $t$ errors in total, within fault-tolerance considerations. This simple repetition scheme, however, is not optimal; in App.~\ref{app:ft_syndrome}, we discuss various methods to reduce the number of repetitions required, thus improving the efficiency of the syndrome measurement process.

\subsubsection{Recovery and decoupling}\label{sec:recovery}
In the recovery step, we apply local $X$ operations to the extended set of data qubits, based on the error sequences inferred from the syndromes obtained in the parity-check step (see Sec.~\ref{sec:XXmeas} for the decoding algorithm). Finally, in the last step of the subcircuit, we decouple the coupling ancilla qubits from the original data qubits, using the same circuit as in the coupling step. The coupling ancilla qubits are then measured in the $Z$ basis and the measurement outcomes are recorded for the purpose of determining damped rows later.

\subsubsection{Fault tolerance of a subcircuit}

We now come back to address the fault tolerance of the subcircuit against $\cF_a$ errors: fault tolerance of the subcircuit is ensured by the fault tolerance of its individual components. Assume that there are $s$ incoming $\cF_a$ errors to the subcircuit (both rows), and then $r_1$, $\ldots$, $r_5$, and $r_6$ $\cF_a$ errors in the coupling, damping syndrome extraction, $\overline{XX}$ measurement, parity check, recovery, and decoupling steps, respectively. Also, assume that $s+\sum_{i=1}^{6}r_i \leq t$. We can leverage the fault-tolerance property at every step to show that there are no more than $r_4+r_5+r_6$ $\cF_a$ errors in the outgoing state of the subcircuit, thereby showing that the subcircuit is fault tolerant against $\cF_a$ errors. Specifically:
\begin{itemize}
    \item Since $s+r_1 + r_2\leq t$, the syndrome extraction step (Sec.~\ref{sec:damping}) ensures that its outgoing state has no more than $r_1 + r_2$ $\cF_a$ errors.
	\item Since $r_1 + r_2 + r_3 \leq t$, the $\overline{XX}$-stabilizer measurement step (Sec.~\ref{sec:XXmeas}) ensures that the $X$-type error sequences of its outgoing state have the property stated in Lemma~\ref{lemma:XXmeasurement}.
	\item This, together with the fact that $r_3 + r_4 \leq t$, ensures that correlated errors from the $\overline{XX}$-stabilizer measurement step are detected and its outgoing state differs from the detected state by no more than $r_4$ $\cF_a$ errors (Sec.~\ref{sec:parity_checks}). 
	\item In the recovery step, gates are applied transversally (Sec.~\ref{sec:recovery}). This ensures that there are no more than $(r_4 +r_5)$ $\cF_a$ at the end of this step.
	\item Finally, the decoupling step also applies transversal gates with respect to the data qubits (Sec.~\ref{sec:recovery}). This means there are no more than $(r_4+r_5+r_6)$ $\cF_a$ errors in the outgoing state of the subcircuit.
\end{itemize}

Furthermore, for the discussion on decoding $Z$-type errors ($Z$ and $\cF_z$) in the next section, it is important to note that a subcircuit does not spread $Z$-type errors badly. First, a $Z$-type error on any data qubit trivially passes through all steps except for the $\overline{XX}$ measurement, where it propagates to the measurement ancilla. However, the role of this ancilla is limited to the detection of the error and moreover, the propagated error does not propagate back to the data qubits. Second, a $Z$-type error on any coupling ancilla behaves in the same manner as on a data qubit, except in the decoupling step, where it can propagate back to the data qubits. However, it can only propagate to at most one data qubit. This is also true for $Z$-type errors on measurement ancillas of parity measurements. Overall, the total number of $Z$-type errors in the outgoing state of the subcircuit is thus limited by the number of $Z$-type errors inside the subcircuit.

We note that our construction guarantees that the output of the EC gadget has no more than $(r_4+r_5+r_6)$ errors, more stringent than required by the fault-tolerance condition where the output can have up to $\sum_{i=1}^6r_i$ errors [see property (P1) of App.~\ref{app:ftprops}], equal to the number of errors that originate in the EC gadget itself. The construction can thus be optimized to relax some of the error propagation and correction requirements, but we opted for the current modular structure where components are made individually fault-tolerant, for a simpler argument of fault tolerance of the full EC gadget.

\subsection{Decoding $Z$ errors}
\label{sec:decodingz}
	
As noted earlier, each $\overline{XX}$ measurement between a pair of rows in a single subcircuit outputs only one syndrome bit. To have sufficient information for decoding $Z$ errors, all the $\overline{XX}$ stabilizers between \emph{every} pair of consecutive rows need to be measured. One round of $Z$-syndrome measurement consists of $t$ subcircuits performed in two steps, thereby giving a $t$-bit syndrome sequence. In the first step, the subcircuits are performed between rows $1 \, \& \, 2$, $3 \, \& \, 4$, and so on, as shown in Fig.~\ref{fig:ECBigPicture}(a). In the second step, they are performed between rows $2 \, \& \, 3$, $4 \, \& \, 5$, and so on. Note that if there is an odd number of rows in the lattice, the last row stays idle in the first step, and the first row stays idle in the second step. If the number of rows is even, the first and the last rows stay idle in the second step.  
	
Due to faults in the $\overline{XX}$ measurements, the $t$-bit syndrome obtained after the first round of measurements is unreliable. We therefore perform many rounds of the $\overline{XX}$ measurements, using the same repetition scheme as in the parity-check step of a subcircuit. Specifically, we repeat until a syndrome sequence appears $t+1$ times and use it to decode the $Z$ error. 
	
In addition to the $\overline{XX}$ syndrome, we also need information about the damped/undamped rows from the damping syndrome extraction step to complete the decoding, as explained in App.~\ref{app:err_corr}. In the following, we explain how to get this information from the various syndrome bits obtained in the subcircuits. 

Consider a subcircuit and a row involved in that subcircuit. 
\begin{itemize}
	\item \emph{Damping syndrome extraction:} An outcome of $-1$ in any of the $ZZ$-stabilizer measurements or the individual $Z$ measurements done on the row clearly indicates a damping error in the incoming state to the subcircuit or a damping fault in the syndrome extraction unit. In either case, we label the row as damped. 
	\item \emph{Parity checks}: An outcome of $-1$ in any of the parity measurements done on the row indicates a damping fault somewhere in the damping syndrome extraction, the $\overline{XX}$ measurement, or the parity-check step. If the flag qubit in the $\overline{XX}$ measurement yields an outcome of $+1$, then the row is genuinely damped or there are at least two damping errors on both the ancillary and flag qubits in the $\overline{XX}$ measurement. Regardless, we label the row as damped in this case. On the other hand, if the flag qubit yields an outcome of $-1$, then it is possible that errors detected by the parity measurements are merely due to a damping error on the ancillary qubit in the $\overline{XX}$ measurement. In this case, we label the row as \emph{potentially damped} row.

    Without the flag qubit in the $\overline{XX}$ measurement, a potentially damped row might be incorrectly identified as a truly damped row. For instance, consider a $7\times 7$ Bacon-Shor code capable of correcting up to $6$ damping errors. The extraction of $Z$ syndrome can be ambiguous for the following two scenarios: (a) four damping errors $\cF_a$ each affecting one of the rows $1$, $2$, $3$, and $4$, along with an additional $Z$ error on row $5$, (b) two damping errors on the ancilla qubits measuring $\overline{X}_1\overline{X}_2$ and $\overline{X}_3\overline{X}_4$, together with two $Z$ errors on rows $6$ and $7$. The errors on the ancilla qubits can be arranged such that $X$ errors appear on the first four rows. Both scenarios can thus produce an identical syndrome of $(+1, +1, +1, +1, -1, +1)$ for $Z$ errors. For scenario (a), we need to apply $Z$s to the first five rows, whereas for scenario (b), we need to apply $Z$s to the last two rows; the flag qubit allows us to distinguish between the two.
		
	Note that there are cases where a coupling ancilla is damped and the subcircuit under consideration detects and corrects the $X$ part of the damping error. The ($I-Z$)-part, however, is transferred to the data qubits after the decoupling step and enters the next subcircuit. None of the parity measurements in the succeeding subcircuit will detect this error. However, we must be aware that a $Z$ error, if projected out by the $\overline{XX}$ measurement in the next subcircuit, comes from a first-order damping error. To take this into account, we also declare a row to be damped/potentially damped if the preceding subcircuit has labelled it as damped/potentially damped.

	\item \emph{Decoupling:} All the coupling ancillas are measured in the $Z$ basis after being decoupled from the data qubits. If any of them yields a $-1$ outcome, it indicates a damping fault after the parity-check step. The fault may be detected by the damping syndrome extraction in the next subcircuit if the fault is on a data qubit, or undetected by the next subcircuit if the fault is on a coupling ancilla (because only the $(I-Z)$-part remains, as explained above). So, a row in the current subcircuit is also labeled as damped if any of the coupling ancilla measurements in the previous subcircuit has outcome $-1$.  
\end{itemize} 
In any \textsc{EC} unit, a single row of data qubits is  a part of many subcircuits. We say that the row is damped/potentially damped if any of the subcircuits in the \textsc{EC} unit flags it as damped/potentially damped. Note that a potentially damped row can later be marked as a damped row if a condition for damped row is matched. However, once the status of a row is marked as damped, it will not be changed to potentially damped or undamped (for the current EC unit).

Once the $\overline{XX}$ syndromes and the information about damped/potentially damped/undamped rows are collected, we decode and correct for $Z$ errors using the following procedure. Recall from~\ref{sec:XXmeas} that there are a pair of (complementary) error sequences $e$ and $\bar{e}$ compatible with a given $\overline{XX}$ syndrome. Let $n_d, n_p, n_u$ denote the number of damped, potentially damped, and undamped rows, respectively. Let $z_d(e), z_p(e), z_u(e)$ denote the number of $Z$ errors in the damped, potentially damped, and undamped rows for the sequence $e$. Then, the number of $Z$ errors in the damped, potentially damped, and undamped rows for the sequence $\bar{e}$ are $z_d(\bar{e}) = n_d - z_d(e)$, $z_p(\bar{e}) = n_p - z_p(e)$, and $z_u(\bar{e}) = n_u - z_u(e)$.

Next, we determine the more probable error sequence by summing up the number of $Z$ errors with suitable weights. We assign a weight $w(e;n_d,n_p,n_u) = n_d + n_p + z_p(e) + 2z_u(e)$, to the error sequence $e$ with respect to the configuration of $n_d$ damped rows, $n_p$ potentially damped rows, $n_u$ undamped rows. This definition is motivated by the relative orders of the different errors, in terms of the noise strength $p$. Specifically, a damped row requires at least a first-order damping event to occur, thus $n_d$ damped rows implies an event of at least $O(p^{n_d})$. Similarly, $n_p$ potentially damped rows arise only when an event of order at least $O(p^{n_p})$ takes place. Finally, a $Z$ error found in a potentially damped row can result from a first-order event (for example, a true damping error) and thus, $z_p(e)$ is added to the weight. On the other hand, a $Z$ error appearing in an undamped row necessitates at least a second-order event (for example, a genuine $Z$ error) which leads to the inclusion of $2z_u(e)$ in the weight.

Now, between $e$ and $\bar{e}$, we choose the sequence with the lower weight as the output of the decoding. This is equivalent to comparing $z_p(e)+2z_u(e)$ and $z_p(\bar{e})+2z_u(\bar{e})$. If there is no potentially damped row (i.e. $z_p=0$), it reduces to comparing the number of $Z$ errors in the undamped rows, as what is done for the ideal $\textsc{EC}$ (see App.~\ref{app:err_corr}).

This completes our description of the fault-tolerant \textsc{EC} gadget. Further optimizations of the EC gadget are possible. In particular, we note that there are multiple places in the gadget where we need to repeat the measurements until we obtain $t+1$ identical repeats of the syndrome sequences. A simple counting gave a worst-case $t(t+1)+1$ repetitions of the requisite measurement. A more optimal approach gives a reduced requirement of no more than $\frac{1}{4}(t+2)^2+1$ repetitions, fewer than $t(t+1)+1$ for $t> 1$; see App.~\ref{app:ft_syndrome} for further details.

\subsection{Fault Tolerance of the EC Gadget\label{sec:FTSummary}}
We now pull together all the arguments that explain why the EC gadget satisfies the fault-tolerance property (P1) of App.~\ref{app:ftprops}. As mentioned at the end of Sec.~\ref{sec:subcircuit}, a subcircuit is fault tolerant against $\cF_a$ errors. Since an EC unit consists of multiple subcircuits, it is also fault tolerant against $\cF_a$ errors. 

In addition, a subcircuit contains the spread of $Z$ errors. Hence, at the end of an EC unit, the outgoing state has no more $Z$ errors than the number of errors entering and occurring inside the unit. To be fault tolerant against $Z$ errors, however, we need to be sure that decoding the $Z$ errors does not lead to a logical error. This amounts to having reliable information of the $Z$-type syndrome and the damped rows. The $Z$-type syndrome is reliable due to the repeat-until-success scheme we adopt (see App.~\ref{app:ft_syndrome}) and all damped rows are flagged according to the protocol described in the last section.

\section{Logical gadgets}\label{sec:czgadget}

To complete our discussion of fault-tolerant computation with Bacon-Shor codes, we describe how to construct a universal set of logical gadgets fault tolerant against amplitude-damping noise. Our construction of the single-qubit logical gadgets, as well as the preparation and measurement gadgets, follows the prescription of our earlier work \cite{jayashankar2022achieving} for the $4$-qubit code (the smallest Bacon-Shor code); see Apps.~\ref{sec:ftgadgets} and \ref{sec:universal} for a self-contained description. We hence confine our discussion here to the logical \textsc{CZ} gadget, as it presents some interesting new features. We note that one can also construct a $\textsc{CCZ}$ gadget, adding to the variety of multi-logical-qubit gates; see App.~\ref{sec:ccz}.

The logical \textsc{CZ} gate for the Bacon-Shor codes can be realized transversally by performing physical \textsc{CZ} gates between qubits in each row of one code block and those in each column of the other code block. A specific implementation---there are many that work---is to have the physical \textsc{CZ} gate between qubit $(i,j)$ of one code block and qubit $(j,i)$ of another code block, as illustrated in Fig.~\ref{fig:CZ2} for the $3\times 3$ code. That this works as the logical \textsc{CZ} gate can be understood by noting that the logical states $|0\rangle_L$ and $|1\rangle_L$ can be written as
\begin{equation}\label{eq:01codewords}
	{\left.\begin{array}{c}|0\rangle_L\\|1\rangle_L\end{array}\right\}}={\left(\tfrac{1}{\sqrt 2}\right)}^{n+1}\!\!\!\!\!\!\!\sum_{x\in\{0_r,1_r\}^n}{\left[1\pm (-1)^{|x|}\right]}|x\rangle,
\end{equation}
where $x$ is an $n$-bit string, $|x|$ counts the number of $1$s in $x$ (the number of $1_\textrm{row}$s in $x$), and $|0(1)_r\rangle\equiv |0(1)\rangle_\textrm{row}\equiv |0(1)\rangle^{\otimes n}$ from Eq.~\eqref{eq:codewords}. The individual bits of $x$ thus go over the different rows of qubits of the Bacon-Shor lattice, specifying if each row of qubits is in the $|0\rangle_\textrm{row}$ or the $|1\rangle_\textrm{row}$ state. Observe that $|0\rangle_L$ is a sum of terms, each of which has an even number of rows in $|1\rangle_\textrm{row}$, i.e., $|x|$ is even; $|1\rangle_L$ is composed only of terms with odd $|x|$. Since the logical $Z$ is equivalent to applying $Z$ to any column, if the state is $|0(1)\rangle_L$, an even(odd) number of columns would have had a logical $Z$ applied to them, amounting to an overall identity(logical $Z$) operation, thus accomplishing the logical \textsc{CZ} gate.

\begin{figure}[t]
	\centering
	\includegraphics[trim=75mm 80mm 140mm 10mm, clip, scale=0.4]{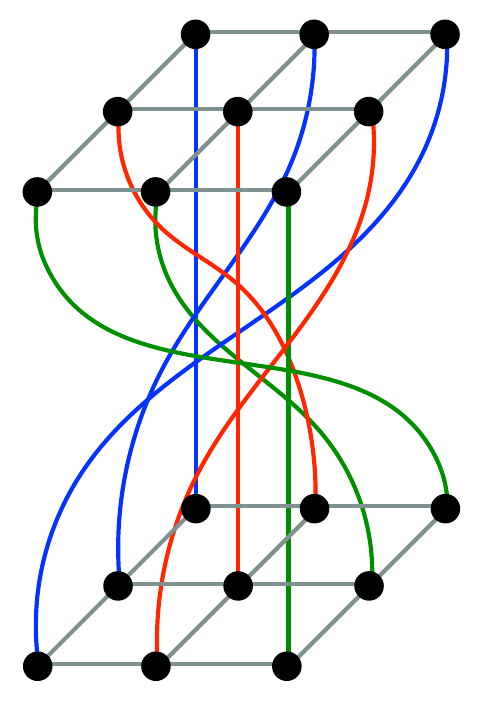}
	\caption{Transversal implementation of the encoded \textsc{CZ} gate for $3\times3$ Bacon-Shor code. The physical \textsc{CZ} gates are between the pairs of qubits, one from the top code block, the other from the bottom code block, connected by a colored (red/blue/green) line. The colors mark the qubits in the same row/column for the top/bottom code blocks.}
	\label{fig:CZ2}
\end{figure}

It is not apparent that this construction is fault tolerant. As emphasized in Ref.~\cite{jayashankar2022achieving}, transversality in such a noise-adapted situation does not automatically mean fault tolerance, as the transversal gates can still convert a correctable error into an uncorrectable one. Now, an $\cF_z$ or a $Z$ error passes through the logical \textsc{CZ} without change and the transversality does assure fault tolerance. An $\cF_a$ error, however, in the incoming state of, say, the top code block will propagate a $Z$ error to the bottom block under a transversal \textsc{CZ} gate. An \textsc{EC} gadget following the bottom block may not, by itself, correctly decode the $Z$ errors as our \textsc{EC} construction will incorrectly treat these first-order propagated $Z$ errors as the second-order $Z$ errors when trying to decode the detected errors. However, this problem can be resolved by incorporating the information about the damped rows in the first block. Specifically, when decoding $Z$ errors for an output block of a logical \textsc{CZ}, the set of rows regarded as damped (used in the maximum-likelihood decoding of the $Z$-error syndromes described in Sec.~\ref{sec:decodingz}) should include also those rows with qubits connected to damped qubits in the other block (as indicated by the damping syndrome extraction in the \textsc{EC} gadget of the other block). With this joint decoding of the two code blocks, we achieve fault tolerance of the logical \textsc{CZ} gadget.

\section{Logical infidelity and pseudothreshold for quantum memory}\label{sec:threshold}

To demonstrate the fault-tolerance of our EC gadget construction, here we estimate the logical infidelity for different code sizes $L=t+1$ and compute the pseudothreshold levels for the basic task of quantum memory.

Following the standard fault-tolerance analysis of Ref.~\cite{aliferis}, the pseudothreshold can be estimated by considering an ``extended memory gadget" comprising an idling (that is, memory) step sandwiched between two \textsc{EC} units. 
Specifically, the pseudothreshold is that critical value of the noise strength, $p=p_{\rm th}$, below which the  fault-tolerant quantum memory outperforms the unencoded (i.e., not error-corrected) version. In particular, we compare the infidelity with and without error correction. For amplitude-damping noise with strength $p$, the fidelity of the unencoded qubit averaged over all input states falls as $1 -\alpha p$ for small $p$, where $\alpha=1/3$. We write $1-F(p)$ for the infidelity between the encoded output and the input states, with $F(p)$ a function of $p$. Then, the pseudothreshold $p_\mathrm{th}$ is determined by the condition,
\begin{equation} 
	1-F (p_\mathrm{th}) = \alpha p_\mathrm{th}.
\end{equation}

Often, the pseudothreshold for a fault-tolerant quantum computing scheme is computed by exact counting or full computer simulation, for a close estimation. As we are after only a simple gauge of the efficacy of our scheme, we opt instead for an approximate bound by counting the number of ways $t+1$ errors can occur in the extended memory gadget. The $t+1$-error term gives the dominant contributor to the infidelity. Specifically, for a given noise strength $p$, we estimate, 
\begin{equation}\label{eq:thbound}
	1-F(p) < \left[{2N+(t+1)^2 \choose t+1} - {N \choose t+1}\right]p^{t+1}.
\end{equation}
Here, $N$ is the total number of circuit locations in our \textsc{EC} gadget, and the second term in the brackets excludes the ways of having $t+1$ errors in the first \textsc{EC} gadget in the extended memory gadget; those are counted already in the previous extended gadget. 

Note that we have assumed that $p$ is small enough such that larger-than-$t$-error terms give a contribution to the infidelity that is smaller than the overcounting we have here from counting all $(t+1)$-error locations, including those benign to the infidelity (usually quite a significant fraction). 
We need also then an estimate of $N$. We can do a careful counting using our EC-gadget construction (see App.~\ref{sec:ft_pth} for the details), and we find $N=O(t^6)$, a rather rapid scaling with $t$ due in part to the necessity to repeat the parity measurements for fault tolerance. The dominant part in a subcircuit is repeating the parity measurements, which take, in the worst case, $O(t^2)$ rounds with $O(t)$ location for each round. This gives $O(t^3)$ locations for a subcircuit. For the whole EC, there are $O(t)$ subcircuits. In the worst case, $\overline{XX}$ measurements are repeated $O(t^2)$ rounds for reliable decoding $Z$ errors. This gives the $O(t^6)$ number of locations in total.

\begin{figure}[t]
	\centering
	\includegraphics[width=\columnwidth]{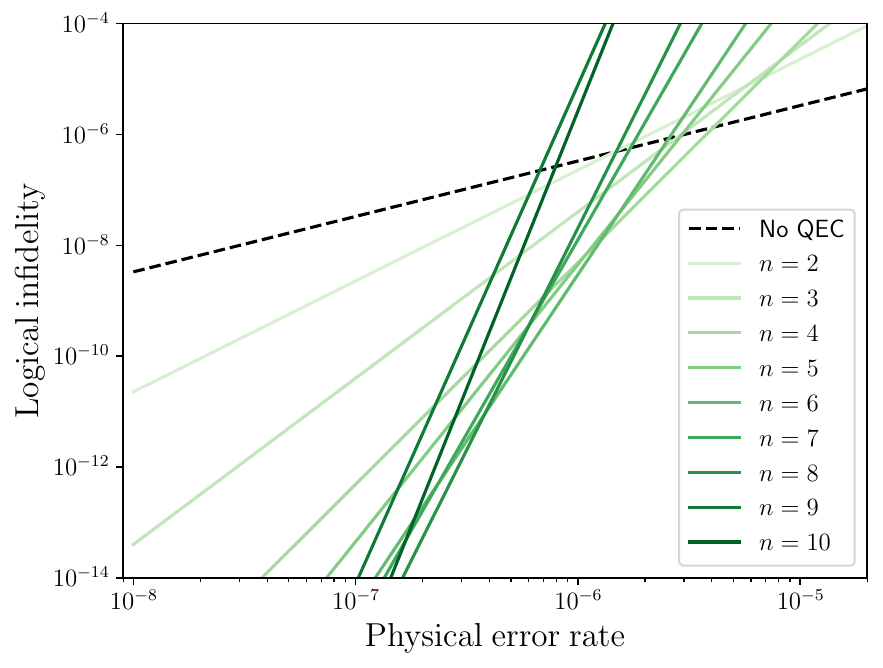}
	\caption{Upper bounds on the infidelity of the output state after the extended memory gadget, plotted as a function of the amplitude-damping parameter $p$ for various lattice sizes $n=t+1$. The dotted line is the average infidelity for the unencoded state, while the green lines of varying shades are for the different lattice sizes $n=2$ to $n=10$, from light to dark green.} 
	\label{fig:thres}
\end{figure}

Fig.~\ref{fig:thres} plots the upper bound on the infidelity for different values of the damping parameter $p$ and for different Bacon-Shor code sizes $n$ (i.e., $n\times n$ lattice, for $n=t+1$). We show also the infidelity without error correction (dotted line), averaged over all pure states. We can then read off lower bounds on the pseudothresholds, i.e., the values of the damping parameter $p$ below which, for a given $n\times n$ code, error correction gives improved infidelity values over the unencoded situation. These pseudothreshold values are given in Table \ref{tab:pseudothres}. We note that, for each damping strength, there is an optimal lattice size that gives minimal infidelity. This is similar to what was observed previously for Bacon-Shor codes for correcting Pauli errors \cite{napp2012optimal}. This suggests that there is no non-vanishing accuracy threshold value for the damping strength below which increasing the code size always lowers the infidelity.

We note that our infidelity estimates, and hence the pseudothreshold values, here are likely far too pessimistic, due to our simple counting approach. For example, in our earlier work on the $n=2$ case \cite{jayashankar2022achieving}, a more careful estimate gave a pseudothreshold of $1.5\times10^{-4}$, compared to the corresponding $\sim 10^{-6}$ value in Table~\ref{tab:pseudothres}. In particular, we count in Eq.~\eqref{eq:thbound} \emph{all} ways of having $t+1$ errors, not just those that do eventually lead to an uncorrectable situation. As is commonly observed in such analyses, a significant fraction of $>t$-error situations are in fact benign, i.e., they give only correctable errors; this is certainly the case in our earlier $n=2$ construction \cite{jayashankar2022achieving}. 

For our current general-$n$ construction, we note that not all fault locations from the leading \textsc{EC} in the extended memory gadget will combine with fault locations in the trailing \textsc{EC} to form malignant sets of faults that do cause uncorrectable errors. Rather, only faults at locations near the end of the leading EC can combine with faults in the trailing EC to cause logical errors; most combinations of faults in the first \textsc{EC} can in fact be corrected. This is true even for subcircuits while considering only $\cF_a$ errors---most of the combinations of $\cF_a$ errors in a subcircuit can be eliminated by the subcircuit itself; only a small portion of faults occurring at the end of a subcircuit can propagate errors to combine with faults from subsequent subcircuits. In addition, during idle time, a qubit cannot be damped more than once, so we should exclude sets with multiple dampings on the same qubit. All these benign situations add to our upper bound on the infidelity in Eq.~\eqref{eq:thbound}, but a more careful counting can exclude them to give a tighter upper infidelity bound and improved pseudothresholds. 

Nevertheless, our goal here is simply to demonstrate the existence of nonvanishing pseudothresholds for our construction, and, more importantly, to show that the infidelity decreases as $n$ increases below the pseudothresholds, as must be the case for fault-tolerant operation.

\begin{table}
	\begin{tabular}{c@{\hspace{0.5cm}}c}
		\hline\hline
		$n$ \ & \ $p_{\text{th}}$ \\
		\hline
		2 & $1.46 \times 10^{-6}$ \\
		3 & $2.88 \times 10^{-6}$ \\
		4 & $4.08 \times 10^{-6}$ \\
		5 & $2.93 \times 10^{-6}$ \\
		6 & $2.57 \times 10^{-6}$ \\
		7 & $1.72 \times 10^{-6}$ \\
		8 & $1.49 \times 10^{-6}$ \\
		9 & $6.73 \times 10^{-7}$ \\
		10 & $7.94 \times 10^{-7}$ \\
		\hline\hline
	\end{tabular}
	\caption{\label{tab:pseudothres} Lower bounds on the pseudothreshold for lattice sizes up to $n=10$.}
\end{table}

\section{Conclusion}\label{sec:conclude}

We have shown how to construct an error-correction gadget as well as a set of universal logical operations based on the Bacon-Shor family of codes, all fault-tolerant against amplitude-damping noise on individual physical components. This completes the work started in Ref.~\cite{jayashankar2022achieving}, giving a fully fault-tolerant approach for quantum computation capable of correcting amplitude-damping errors on multiple qubits. Our work demonstrates the possibility of fault-tolerant construction starting only from the noise description, without having to adapt the scheme from a more noise-agnostic protocol.

Due to the adaptation to a specific noise channel, rather than having a full error basis, guaranteeing fault tolerance here is not just about controlling the spread of errors, but also controlling the \emph{types} of errors that emerge from the gates---they have to remain correctable by our noise-adapted code. As noted already in our earlier work, this is the basic reason why transversal constructions of operations, a cornerstone of general-purpose FT schemes, do not automatically assure fault tolerance in noise-adapted situations. Transversality guarantees a low spread of errors, but the types of errors that emerge are entirely determined by the gates used in the gadget construction and how the errors propagate through them. 

In our earlier work, we dealt with this issue by momentarily extending our four-qubit code to an eight-qubit one that could deal with additional varieties of errors beyond amplitude-damping ones, before switching back to the original four-qubit code. Here, we extend that approach to general Bacon-Shor codes, with the use of the coupling ancillas and the $(2t+1)$-qubit repetition code in the EC gadget. In addition, we present a curious solution to this problem in the CZ gadget -- we deal with the originally uncorrectable extra $Z$ error, propagated to the second code block due to an amplitude-damping error in the first code block, by a joint decoding of the two code blocks to recognize this propagated $Z$ error. Such a joint decoding to extend the fault-tolerance capabilities, requiring no joint quantum measurements but only pooling of the classical syndrome information, is likely useful in other contexts.

A different approach to achieving fault tolerance in such noise-adapted contexts would be to employ only gates that are ``noise-structure-preserving", so that incoming errors arising from the noise channel remain among the correctable error types after propagation through the gate. Such noise-structure-preserving gates have been identified for Pauli errors in the past~\cite{aliferis_biased, bias_preserving2020}. For amplitude-damping noise, however, we have not been able to find a set of noise-structure-preserving physical gates with sufficient variety to construct a universal set of logical operations.

Of course, our gadget constructions here are far from optimal. The simplistic repetition of the syndrome checks, reminiscent of the repeated syndrome rounds of surface-code error correction, is but an easy-to-analyze route to ward against errors in the syndrome measurements; App.~\ref{app:ft_syndrome} offer some optimizations but certainly more can be done. Our component-by-component approach to assuring the fault tolerance of the EC gadget gave a clear path to showing the fault-tolerance properties, but led to fewer output errors than required by FT rules; a more holistic construction no doubt can lower resource demands and thereby improve the pseudothreshold numbers and truly realise the potential of noise adaptation. 

Nevertheless, we have shown the possibility of a fully fault-tolerant quantum computation scheme adapted for amplitude-damping noise. We hope our work will motivate exploration into schemes adapted for other types of noise, including ones that are directly relevant to current quantum computing experiments.

\section*{Acknowledgments}
This work is supported in part by the Ministry of Education, Singapore (through Grant No. MOE2018-T2-2-142). HK Ng acknowledges partial support from the National Research Foundation, Singapore and A*STAR under its CQT Bridging Grant. P.M. is supported in part by a grant from the Mphasis F1 Foundation to the Centre for Quantum Information, Communication, and Computing (CQuICC).

\bibliography{FTScalingUp}

\appendix
\section{Correcting amplitude-damping noise with the Bacon-Shor Code}\label{app:err_corr}

Here, we give further details of the ideal EC procedure of Ref.~\cite{piedrafita2017reliable} to correct amplitude-damping noise using the Bacon-Shor code and illustrate its use with the $3\times 3$ example.

\subsection{Ideal EC procedure}\label{sec:ideal}
The ideal EC procedure constructed in Ref.~\cite{piedrafita2017reliable} involves three steps performed in sequence: a two-part syndrome extraction for damping errors, a $Z$-error syndrome extraction, and a recovery step. We explain each step in detail below.

\begin{enumerate}
\item \textit{Damping error syndrome extraction}\\[-4ex]
\begin{enumerate}[{(a)}]
		\item Measure the $Z_{i,j}Z_{i,j+1}$ stabilizers, for all $i$ (row) and $j$ (column). For each measurement, the two qubits involved are said to have even(odd) parity if the outcome is $+1(-1)$. In each row, if at least one parity outcome is odd, we declare that the row is damped; otherwise, the row is undamped. 
		\item For each damped row, we measure all the qubits in the row in the $Z$ basis to determine the positions of the damping errors. If the measurement outcome is $+1$, the qubit has a damping error; if the outcome is $-1$, the qubit is not damped. This can be seen by considering the effect of a damping error $E$ on an arbitrary encoded state: $E$ has effect only if the qubit is originally in the $\ket 1$ state; the damped row is disentangled from the rest of the rows, regardless of the original encoded state, and the damped qubit ends up in the state $\ket{0}$.  
\end{enumerate}
	
We note that the above steps can identify the positions of up to $t=n-1$ damping errors in each row \footnote{Note that the code can actually tolerate up to $n-1$ damping errors in each row, for any $n-1$ rows, i.e., maximally $(n-1)^2$ damping errors. However, it can only tolerate $n-1$ damping errors in \emph{arbitrary} locations, as there are cases where $n$ damping errors cause a logical error, e.g., $n$ errors in one row or one column.)}. At the end of this syndrome extraction step, we are ready to correct for the $X$-part of the damping error $E=X(I-Z)$ by applying $X$ to the damping locations, leaving $(I-Z)$-part for the subsequent detection. Alternatively, we can delay the application of these local $X$s until the end of the EC procedure. This will not affect the detection of $ Z$ errors in the next part.
	
\item \textit{$Z$-error syndrome extraction}\\[0.5ex]
Measure all the $\overline{X}_i\overline{X}_{i+1}$ stabilizers and obtain the syndrome bits. An $\overline{XX}$ measurement on undamped rows has the effect of detecting $Z$ errors. 
If one of the rows is damped, as detected in Step 1(a), the $\overline{XX}$ measurement either removes that damping error and returns the correct code state (if outcome $+1$ is obtained), or, it projects the qubits onto the correct code state up to a $Z$ error on one of the qubits in the damped row (if outcome $-1$ is obtained). The measurement of $\overline{XX}$ stabilizers can result in different post-measurement errors depending on whether any of the rows involved were damped prior to the measurement. Table~\ref{tab:XXmeasFa} gives the scenarios that can occur for the two rows before the measurement, along with the measurement outcomes and the subsequent errors.
	
\begin{table}
		\centering
		\begin{tabular}{c|c|c|c}
			Damped rows & Initial error & Outcome  & Final error  \\
			\hline
			None & $\overline{I}_1\overline{I}_2$ & $+1$ & $\overline{I}_1\overline{I}_2$ \\
			None & $\overline{Z}_1\overline{I}_2$ & $-1$ & $\overline{Z}_1\overline{I}_2$ \\
			1 & $(\overline{I}_1-\overline{Z}_1)\overline{I}_2$ & $+1$ & $\overline{I}_1\overline{I}_2$ \\
			1 & $(\overline{I}_1-\overline{Z}_1)\overline{I}_2$ & $-1$ & $\overline{Z}_1\overline{I}_2$ \\
			1,2 & $(\overline{I}_1-\overline{Z}_1)(\overline{I}_2-\overline{Z}_2)$ & $+1$ & $\overline{I}_1\overline{I}_2$ or $\overline{Z}_1\overline{Z}_2$ \\
			1,2& $(\overline{I}_1-\overline{Z}_1)(\overline{I}_2-\overline{Z}_2)$ & $-1$ & $\overline{Z}_1\overline{I}_2$ or $\overline{I}_1\overline{Z}_2$ \\
		\end{tabular}
		\caption{List of possible syndromes and the projected errors after measuring an $\overline{XX}$ stabilizer in the presence of damping errors. The two rows involved are denoted as $1$ and $2$. $\overline{I}_1$ ($\overline{I}_2$) denotes all identities in row $1$(row $2$), $\overline{Z}_1$($\overline{Z}_2$) denotes a single $Z$ operator on one of the qubits in row 1(row 2).	It is assumed that the $X$-part of the damping error is corrected, leaving only the $(I-Z)$-part.}
		\label{tab:XXmeasFa}
\end{table}

Measuring the $\overline{XX}$ stabilizers also has the effect of annihilating the off-diagonal $\cF_z$ errors. To see this, note that $\cF_z(\cdot)\equiv \tfrac{1}{2}[(\cdot)Z + Z(\cdot)]$ and 
\begin{equation}
	(1+\overline{X}_i\overline{X}_{i+1}) (Z\rho + \rho Z) (1+\overline{X}_i\overline{X}_{i+1}) = 0
\end{equation}
where $\rho$ is a state in the code space and the $\cF_z$ error is on row $i$ or $i+1$. However, an even number of $\cF_{z}$ errors may not be eliminated by an $\overline{XX}$ measurement. For example, if there are two $\cF_z$'s on a row, it has the effect of a single $Z$ error, since
\begin{equation}
	\cF_{z} \otimes \cF_{z}(\cdot)=\frac{1}{2}[(\cdot) + Z(\cdot)Z].
\end{equation}
The double $\cF_z$ errors are not annihilated but contribute to the probability of measuring $Z$ errors. We need to consider two cases: (i) even numbers of $\cF_{z}$'s on both the rows and (ii) odd numbers of $\cF_{z}$'s on both the rows. In case (i), the post-measurement errors are $Z$ errors and would be corrected as if they originated from genuine $Z$ errors. However, in case (ii), the post-measurement errors still have the off-diagonal form (see Table~\ref{tab:XXmeasFz}). Nevertheless, these errors are eliminated after considering all the $\overline{XX}$ stabilizer measurements. The reason is that there is at least one row with no $\cF_{z}$ error, given that there are no more than $t=n-1$ damping errors, and the measurement involving this row cannot result in errors of the off-diagonal form. 	
	
\item \textit{Recovery}\\[0.5ex] 
As mentioned earlier, the $X$ part of the damping error can be corrected using the information about the damping locations obtained from the damping-error syndrome extraction step. Meanwhile, decoding $Z$ errors involves both the syndromes obtained from the $\overline{XX}$ measurements and the information about the damped rows obtained from the syndrome extraction unit. As can be seen from Table~\ref{tab:XXmeasFa}, $Z$ errors can be projected out from damping errors, but those are first-order errors and should be distinguished from genuine second-order $Z$ errors. 
	
To understand how we can decode the $Z$ error, it is important to note that the measurement of $\overline{X}_i\overline{X}_{i+1}$ fixes the parity of the $Z$ error for the $i$-th and $(i+1)$-th rows. If $\overline{X}_i\overline{X}_{i+1} = +1$, it means that both rows either have no error or both have $Z$ errors. Conversely, if $\overline{X}_i\overline{X}_{i+1} = -1$, only one of them has a $Z$ error. Starting from the first row, the $\overline{XX}$ syndromes will determine the errors for all the other rows. Since there are two possibilities for the $Z$ error in the first row---either no error or one $Z$ error---there are two bit strings compatible with the obtained syndromes, denoted as $f$ and $\neg f$. Since each damped row carries a weight of at least $p$, the determining factor for which sequence is more likely is actually the number of undamped rows having a $Z$ error. Let $f'$ be the sequence (between $f$ and $\neg f$) with a smaller number of $1$s \emph{among the undamped rows}. We correct $Z$ errors by applying a $Z$ operation to the $i$-th qubit if $f_{i}'=1$. The specific qubit within a row to which the $Z$ operation is applied does not matter because they differ only by a stabilizer operation.   
\end{enumerate} 

\begin{table}[t]
	\centering
	\begin{tabular}{c|c|c|c|c}
		Case & Row 1 & Row 2 & Outcome & Final error  \\
		\hline
		1 & Even & Even & $+1$ & $\overline{I}_1\overline{I}_2$ or $\overline{Z}_1\overline{Z}_2$ \\
		2 & Even & Even & $-1$ & $\overline{Z}_1\overline{I}_2$ or $\overline{I}_1\overline{Z}_2$ \\
		3 & Odd & Even & $\pm1$ & Eliminated \\
		4 & Even & Odd & $\pm1$ & Eliminated \\
		5 & Odd & Odd & $+1$ & $\overline{Z}_1\overline{Z}_2(\cdot)+(\cdot)\overline{Z}_1\overline{Z}_2$ \\
		6 & Odd & Odd & $-1$ & $\overline{Z}_1(\cdot)\overline{Z}_2+\overline{Z}_2(\cdot)\overline{Z}_1$
		
	\end{tabular}
	\caption{List of possible syndromes and the projected errors after measuring an $\overline{XX}$ stabilizer in the presence of $\cF_z$ errors. The two rows involved are denoted as $1$ and $2$. The table shows various scenarios: even number of $\cF_z$s in both rows (cases $1,2$), even number in one row and odd number in the other row (cases $3,4$), and odd number in both rows (cases $5,6$). $\overline{I}_1$($\overline{I}_2$) denotes no $Z$ error in row $1$($2$), while $\overline{Z}_1$($\overline{Z}_2$) denotes one $Z$ error on one of the qubits in row 1(2).}
	\label{tab:XXmeasFz}
\end{table}

\subsection{Example: The $3\times 3$ Bacon-Shor Code}
\label{app:BS3}
To help the reader better understand the ideal EC procedure describe above, we explore the example of the $3\times 3$ Bacon-Shor code, capable of protecting against $t=2$ amplitude-damping errors. We label the qubits in the $3\times 3$ lattice as $q_{ij}$, for $i,j=1,2,3$, with $i(j)$ as the row(column) index. The set of stabilizer generators are
\begin{align}
	\cS = \{&Z_{11}Z_{12}, Z_{12}Z_{13}; Z_{21}Z_{22}, Z_{22}Z_{23}; Z_{31}Z_{32}, Z_{32}Z_{33}; \nonumber\\
	&\overline{X}_1\overline{X}_2, \overline{X}_2\overline{X}_3\},
\end{align}
and we have the logical states written explicitly as
\begin{align}
	\ket{0}_L &= \tfrac{1}{2} (\ket{0_r0_r0_r}+\ket{0_r1_r1_r}+\ket{1_r0_r1_r}+\ket{1_r1_r0_r})\nonumber\\
	\ket{1}_L &= \tfrac{1}{2} (\ket{1_r1_r1_r}+\ket{1_r0_r0_r}+\ket{0_r1_r0_r}+\ket{0_r0_r1_r}),
\end{align}
where $\ket{0}_r \equiv \ket{000}$ and $\ket{1}_r \equiv \ket{111}$ denote the $3$-qubit states for each row, and $\ket{\cdot_r\cdot_r\cdot_r}$ are for the three rows. The $\ket{\pm}_L$ states can be written as
\begin{equation}
	\ket{\pm}_L = \frac{1}{2\sqrt{2}} (\ket{0}_r \pm \ket{1}_r)^{\otimes 3}.
\end{equation}

We now describe the three steps of the ideal EC procedure for this $3\times3$ example.

\smallskip
\noindent(1) \underline{Damping-error syndrome extraction}.\\[1ex]
We measure the six parity check operators (two per row) $\{Z_{11}Z_{12}, Z_{12}Z_{13}; Z_{21}Z_{22}, Z_{22}Z_{23}; Z_{31}Z_{32}, Z_{32}Z_{33}\}$. For each row, if at least one of the measurement outcomes is $-1$, we declare that the row is damped and then measure all three qubits in that row in the $Z$ basis to determine the locations of damping errors. The row is otherwise undamped. 

As an example, consider the case where there is one damping error on $q_{11}$ versus the case with damping errors on both $q_{12}$ and $q_{13}$. In both cases, $(Z_{11}Z_{12}, Z_{12}Z_{13}) = (+1, -1)$, but in the first case, $(Z_{11}, Z_{12}, Z_{13})=(-1,+1,+1)$ while in the second case $(Z_{11}, Z_{12}, Z_{13})=(+1,-1,-1)$. Measuring $Z$s on the individual qubits on the damped row thus locates the damping errors.

Once the damping locations have been identified, we are ready to correct for the $X$-part of the damping error $E=X(I-Z)$ by applying $X$ to the damping locations, leaving the $(I-Z)$-part for the subsequent steps. Alternatively, we can delay the application of this local $ X$s until the end of the EC procedure. This will not affect the detection of $ Z$ errors in the next step. 

\smallskip
\noindent(2) \underline{$Z$-error syndrome extraction}.\\[1ex]
We next measure the two check operators $\{\overline{X}_1\overline{X}_2, \overline{X}_2\overline{X}_3\}$. The outcome $\overline{X}_1\overline{X}_2 = +1 (-1)$ means that the parity, in the $\overline{X}$ basis, of the first and second rows are the same (different). Table ~\ref{tab:XXsyndrome} shows the different possible syndrome outcomes for $\overline{X}_1\overline{X}_2$ and the corresponding errors. The same holds for $\overline{X}_2\overline{X}_3$.

\begin{table}[h]
	\centering
	\begin{tabular}{c|c|c|c}
		Initial error & Example & $X$ Syn & Projected error  \\
		\hline\hline
		No damped row & $\overline{I_1}\overline{I_2}$ & $+1$ & $\overline{I_1}\overline{I_2}$ \\
		$Z$ error & $\overline{Z_1}\overline{I_2}$ & $-1$ & $\overline{Z_1}\overline{I_2}$ \\
		$1$ damped row & $(\overline{I_1}-\overline{Z_1})\overline{I_2}$ & $+1$ & $\overline{I_1}\overline{I_2}$ \\
		$1$ damped row & $(\overline{I_1}-\overline{Z_1})\overline{I_2}$ & $-1$ & $\overline{Z_1}\overline{I_2}$ \\
		$2$ damped rows & $(\overline{I_1}-\overline{Z_1})(\overline{I_2}-\overline{Z_2})$ & $+1$ & $\overline{I_1}\overline{I_2}$ or  $\overline{Z_1}\overline{Z_2}$\\
		$2$ damped rows & $(\overline{I_1}-\overline{Z_1})(\overline{I_2}-\overline{Z_2})$ & $-1$ & $\overline{Z_1}\overline{I_2}$ or  $\overline{I_1}\overline{Z_2}$\\
	\end{tabular}
	\caption{List of possible outcomes and the corresponding errors for the measurement of $\overline{X}_1\overline{X}_2$ in various scenarios. $\overline{Z}_i$ denotes a $Z$ error on the $i$-th row, regardless of the exact position of the error, because they are equivalent up to a stabilizer operator. }
	\label{tab:XXsyndrome}
\end{table}

In the case of a damping error $\cF_a$, the measurements project the state to the code space or to error spaces with different $Z$ errors. In the case of an off-diagonal $\cF_z$, we can consider the following examples to highlight the discussion of the previous section. 
\begin{itemize}
	\item Odd number of $\cF_z$s on one row: $\cF_z$ on $q_{11}$ introduces the term $Z_{11}(\cdot) + (\cdot)Z_{11}$, eliminated by $\overline{X}_1\overline{X}_2$. 
	\item Even number of $\cF_z$s on one row: Two $\cF_z$s on $q_{11}$ and $q_{12}$ introduce the term $Z_{11}Z_{12}(\cdot) + (\cdot)Z_{11}Z_{12} + Z_{11}(\cdot)Z_{12} + Z_{12}(\cdot)Z_{11} \propto \overline{I_1}(\cdot)\overline{I_1} + \overline{Z_1}(\cdot)\overline{Z_1}$. This is an incoherent mixture of $\overline{I}_1$ and $\overline{Z}_1$ and gets detected as a $Z$ error. 
	\item  Even number of $\cF_z$s on two rows: One $\cF_z$ on $q_{11}$ and another $\cF_z$ on $q_{21}$ introduce the term $Z_{11}Z_{21}(\cdot) + (\cdot)Z_{11}Z_{21} + Z_{11}(\cdot)Z_{21} + Z_{21}(\cdot)Z_{11}$. The measurement of $\overline{X}_1\overline{X}_2$ with the outcome $+1$ ($-1$) projects out the term $Z_{11}Z_{21}(\cdot) + (\cdot)Z_{11}Z_{21}$ ($Z_{11}(\cdot)Z_{21} + Z_{21}(\cdot)Z_{11}$). Now, there are off-diagonal terms remaining, but the measurement $\overline{X}_2\overline{X}_3$ eliminates these terms as it sees only one $\cF_z$ on the second row. 
\end{itemize}

\smallskip
\noindent(3) \underline{Decoding and recovery}.\\[1ex]
As mentioned earlier, the $X$-part of the damping error can be corrected using the information about damping locations from the damping-error syndrome extraction step. Meanwhile, decoding for $Z$ error involves both the syndrome obtained from the measurement of $\{\overline{X}_1\overline{X}_2, \overline{X}_2\overline{X}_3\}$ and the information about damped rows obtained from the damping-error syndrome extraction unit. 

To see why the latter is needed, we consider the following two scenarios, both being second-order events. Suppose in both scenarios, we obtain the $Z$-error syndrome $(\overline{X}_1\overline{X}_2, \overline{X}_2\overline{X}_3) = (+1,-1)$, which means the error is either $\overline{Z_1}\overline{Z_2}\overline{I_3}$ or $\overline{I_1}\overline{I_2}\overline{Z_3}$.
\begin{itemize}
	\item Scenario 1: There is one damping error on the first row (any qubit) and another on the second. From Table~\ref{tab:XXsyndrome}, we deduce that $\overline{Z_1}\overline{Z_2}\overline{I_3}$ is the error and there is no way $\overline{I_1}\overline{I_2}\overline{Z_3}$ is projected out in this case.
	\item Scenario 2: There is one second-order $Z$ error on the third row. In this case, the syndrome is always $(\overline{X}_1\overline{X}_2, \overline{X}_2\overline{X}_3) = (+1,-1)$ and the error is $\overline{I_1}\overline{I_2}\overline{Z_3}$. 
\end{itemize}
This example highlights the fact that the $Z$-error syndrome alone is insufficient to distinguish all the correctable errors; the information about the damped rows must also be included. The general rule is that, between two errors compatible with the $Z$-error syndrome, we choose the one with the fewer number of $Z$ errors on the undamped rows.

\section{Fault tolerance properties}\label{app:ftprops}

Here, we state the properties that the EC and logical gadgets must have for fault-tolerant operation against amplitude-damping noise. 

An amplitude-damping channel affecting $n$ qubits locally can be understood as $n$ amplitude-damping errors affecting the data qubits, with $n_1$ first-order $\cF_a$ or $\cF_z$ errors, and $n_2$ second-order $Z$ errors such that $n_1 + 2n_2 \leq n$.  We use the term \emph{fault} to refer to the scenario where noise afflicts a physical location within a gadget and \emph{error} to refer to the resulting change on the quantum state. 

For fault-tolerant operation against amplitude-damping noise, our gadgets must have the following properties (for more details, see Chapter 2 in Ref.~\cite{aliferis2007level}).
\begin{itemize}
	\item[(P1)] An EC gadget is fault tolerant against $t$ errors if, for an input with $s$ errors and an \textsc{EC} gadget with $r$ faults such that $s+r \leq t$, the output state differs from a codeword by at most $r$ errors, and ideally decoding the output state gives the same codeword as ideally decoding the input state. 
	\item[(P2)] A preparation gadget is fault tolerant against $t$ errors if, for $r \leq t$ faults in the gadget, the output differs from a codeword by at most $r$ errors, and ideally decoding the output state gives the same codeword as prepared by an ideal preparation. 
	\item[(P3)] A measurement gadget is fault tolerant against $t$ errors if, for an input with $s$ errors and a measurement gadget with $r$ faults such that $s+r \leq t$, ideally decoding the outcomes gives the same result as ideally decoding and measuring the input state.
	\item[(P4)] A $k$-qubit gate gadget is fault tolerant against $t$ errors if, for input blocks with $s_1, s_2, \dots, s_k$ errors and a gadget with $r$ faults such that $\sum_{i=1}^k s_i + r \leq t$, each output block has at most $s_1+s_2+r$ errors, and, jointly decoding the output blocks by combining their classical syndromes gives the same state as ideal decoding of the input blocks followed by the ideal $k$-qubit operation.  
\end{itemize} 

In the last property, jointly decoding the output blocks implies that the syndrome information from all the blocks can be classically combined to decode correlated errors due to error propagation within the logical gadget. This modification does not affect the fault-tolerance property of the whole circuit and is used when we construct the logical $\textsc{CZ}$ gadget later.  

For fault-tolerant simulation of an ideal circuit, we replace each operation in the ideal circuit by the corresponding logical gadget followed by an \textsc{EC} gadget in each output block (except the measurement gadget which has no output). To discuss the robustness of the simulation circuit, we consider \emph{extended gadgets}. An \emph{extended gadget} is defined as a logical gadget together with the leading and trailing \textsc{EC} gadgets. It can be shown that if there are no more than $t$ faults inside every extended gadget of the simulation circuit, then the properties of the simulation gadgets stated above will ensure that the circuit simulates the ideal circuit correctly. This was proven rigorously in Ref.~\cite{aliferis2007level} for arbitrary local noise and applies to our situation of local amplitude-damping noise. 

In particular, for the memory setting, where we want to store information for a certain time, the argument can be understood as follows. An extended memory gadget includes a waiting step sandwiched between two \textsc{EC} gadgets. Note that the extended gadgets are not separated; rather they overlap: The leading \textsc{EC} of one extended gadget is the trailing \textsc{EC} of the previous extended gadget. Assume there are no more than $t$ faults inside any extended gadget. Consider a particular extended gadget with $s$ errors in the incoming state, $r_1, r_2, r_3$ faults in the leading \textsc{EC}, the waiting step, and the trailing \textsc{EC} respectively, such that $r_1+r_2+r_3 \leq t$. Assume that $s+r_1 \leq t$. Property (P1) then ensures that there are at most $r_1$ errors in the output of the leading \textsc{EC}, thus at most $r_1 + r_2 \leq t$ errors in the input to the trailing \textsc{EC}. Invoking property (P1) one more time, we see that the output of the trailing \textsc{EC} has at most $r_3$ errors. We also see that the incoming state to the following extended gadget, which is the input state to the current trailing \textsc{EC}, has at most $r_1+r_2$ errors. If we now consider the extended gadget that follows, the total number of errors in the incoming state and faults in the leading \textsc{EC} is $r_1+r_2+r_3 \leq t$. This justifies our assumption at the beginning that $s+r_1 \leq t$ since the previous extended gadget behaves in exactly the same way. By stringing together the chain of extended gadgets, we can conclude that the final output state has no more than $t$ errors, thus an ideal decoder attached at the end will decode it correctly.

\section{Proof of Lemma \ref{lemma:XXmeasurement}}
\label{app:ProofThm1}

Here, we provide a proof of Lemma~\ref{lemma:XXmeasurement} stated in the main text.

\smallskip
\noindent \textbf{Lemma \ref{lemma:XXmeasurement}.}
Let $e_i$ and $e_{i+1}$ be the $X$-type error sequences for a pair of successive rows, caused by $\ell \leq t$ errors in the $\overline{XX}$-stabilizer measurement. It is possible to ensure $\textsc{diff}(e_i, e_{i+1}) \leq \ell$ by fixing the order in which the ancillary qubit is coupled to the extended data qubits, as given in Fig.~\ref{fig:XX_stabilizer}.

\begin{proof}
	First, note that a damping error on the flag qubit does not propagate any error to the extended data qubits. It can only propagate $Z$-type error to the ancillary qubit and corrupt its measurement outcome.
	
	Next, consider damping errors only in the ancillary qubit. As damping errors beyond the first one have no additional effect on the qubit, we need only consider a single $\cF_a$ error. In general, an $\cF_a$ error on the ancillary qubit can propagate $X$ errors to multiple extended data qubits in a row, possibly causing more than $t$ qubits in that row to be labeled with an $X$. However, we can ensure $\textsc{diff}(e_i, e_{i+1}) \leq 1$ for every pair of rows by a specific choice of the order in which the ancillary qubit is coupled to the extended data qubits. The zig-zag coupling pattern illustrated in Fig.~\ref{fig:XX_stabilizer} suffices. It is easy to see why this works using the example of the $3\times 3$ code in Fig.~\ref{fig:XX_stabilizer}. If the fault on the ancillary qubit occurs right after coupling with qubit $7$, then the corresponding propagated error sequences are (taking $i=1$) $e_1=XXIII$ and $e_2=XXIII$ and hence $\textsc{diff}(e_1, e_2) = 0$. On the other hand, if the fault occurs right after qubit $3$, then $e_1 = XXXII, e_2=XXIII$ and $\textsc{diff}(e_1, e_2)=1$. In general, the zig-zag coupling pattern ensures that a fault in the ancillary qubit propagates $X$ errors to two rows in a zig-zag manner, thus ensuring $\textsc{diff}(e_i, e_{i+1}) \leq 1$.
	
	If, instead, there are faults only in the extended data qubits and no faults in the ancillary qubit, there are at most $\ell$ qubits marked as \emph{X} and it is straightforward to see that $\textsc{diff}(e_i, e_{i+1})\leq \ell$ for any pair of rows.
	
	More generally, in addition to a fault in the ancillary qubit (which gives $\textsc{diff}_1(e_i, e_{i+1}) \leq 1$), we can allow for at most $\ell-1$ faults distributed among the extended data qubits, leading to $\textsc{diff}_2(e_i, e_{i+1})$ $\leq$ $\ell-1$. Putting these together, the \textsc{diff} function for a pair of error sequences satisfies, $\textsc{diff}(e_i, e_{i+1}) \leq \textsc{diff}_1(e_i, e_{i+1}) + \textsc{diff}_2(e_i, e_{i+1}) \leq \ell$.
\end{proof}

\section{{Optimizing }syndrome extraction}\label{app:ft_syndrome}

In Sec.~\ref{sec:ecunit}, we constructed an EC gadget that is fault tolerant against $\leq t$ amplitude-damping faults. There, we did not consider questions of optimality and resource costs. Here, we carry out a more careful analysis that gives an improved (lower) number of repetitive syndrome measurements needed for fault-tolerant operation. 

We recall that, in the EC gadget, there are three syndrome extraction steps that have to be repeated:
\begin{itemize}
	\item[(i)] \emph{Damping error syndromes:} We repeat the parity measurements on the data qubits of a row at most $t$ times to verify that the trivial syndrome sequence is correctly diagnosed. 
	\item[(ii)] \emph{$\overline{XX}$ measurements:} In the parity-check step after the $\overline{XX}$ measurement, we repeat the parity measurements on the extended data qubits of a row until a syndrome sequence appears $t+1$ times.
	\item[(iii)] \emph{Decoding $Z$ errors:} While decoding the $Z$ errors, we repeat the $Z$-syndrome measurement until a syndrome sequence appears $t+1$ times. 
\end{itemize}
All three steps above can be cast into the problem of fault-tolerant syndrome measurement for the repetition code. In steps (i) and (ii), the repetition code is defined on a row of qubits with $ZZ$ stabilizers between two consecutive qubits; in step (iii), the repetition code is defined across the rows of the lattice with $\overline{XX}$ stabilizers between two consecutive rows. Also, in step (i), we have to verify that the syndrome sequence is indeed the trivial syndrome (all $+1$s), whereas, in steps (ii) and (iii), we have to verify an \textit{a priori} unknown syndrome sequence.

For step (i), if the input state has damping errors but a trivial syndrome sequence is obtained, there must be at least one fault in that round of measurements. Repeating $t$ times, while continuing to obtain a trivial syndrome sequence, means that the last round must be free of measurement faults, and damping errors in the input can be detected since there is at least one error in the input and there are $t-1$ faults in the first $t-1$ rounds; if this is not the case, we would have had $>t$ faults, beyond the requirements of fault tolerance.

For steps (ii) and (iii), we prove the following estimate on the required number of syndrome extraction rounds:
\begin{lemma}[Fault-tolerant syndrome extraction]\label{thm:ft_syndrome}
	To ensure fault tolerance of the \textsc{EC} gadget, it is sufficient to repeat each of the steps (ii) and (iii) no more than $\tfrac{1}{4}(t+2)^2 + 1$ times. 
\end{lemma}

\begin{proof}
	We want to estimate the minimum number of rounds of syndrome extraction required to make the EC gadget fault tolerant. For step (ii) or (iii), the true syndrome sequence is unknown and we choose a sequence based on majority voting. If a sequence appears $t+1$ times, it is safe to decode based on that sequence as there must be at least one round without faults, or we violate the $\leq t$ faults pre-condition for fault tolerance. We imagine a worst-case scenario where a syndrome sequence appears $t$ times before a fault happens and changes the syndrome to another one. This can repeat $t+1$ times as we allow $t$ faults and one no-fault case. Therefore, naively, if we repeat $t(t+1) +1$ times, we are assured that there will be one sequence appearing $t+1$ times. 
	
	However, we notice that there are two types of faults. For the first type, a fault only introduces new errors to the data qubits without affecting the measurement outcomes in a round. They are not harmful because even if we decode using the obtained syndrome, the decoded error differs from the true error by no more than the number of faults in that round. For the second type, a fault changes the measurement outcomes in a round with or without introducing new errors to the data qubits. In this case, we are at risk of introducing more errors if using the obtained syndrome to decode. 
	
	We consider a more specific scenario where there are $r_1$ type-$1$ faults and $r_2$ type-$2$ faults such that $r_1+r_2 = r \leq t$. We imagine, in the worst case scenario, a faulty syndrome sequence appears $r_2$ times due to $r_2$ faults of the second type. We have to repeat enough times so that at least one syndrome sequence without faults appears more than $r_2$ times. Following the earlier argument, it suffices to repeat $N_{\text{rep}} = r_2 + r_2(r_1+1) + 1 = -r_2^2 + (r+2)r_2 + 1$ times since we are then assured that there is one sequence without type-$2$ faults appearing $r_2+1$ times. Maximizing $N_{\text{rep}}$ with respect to $r_2$ gives $N_{\text{rep}}^{\text{max}} = \tfrac{1}{4}(r+2)^2 + 1 \leq \tfrac{1}{4}(t+2)^2 + 1$. Therefore, repeating $\tfrac{1}{4}(t+2)^2 + 1$ times and choosing to decode the syndrome that appears most often is sufficient for fault tolerance.     
\end{proof}

We conclude by discussing two further approaches to reduce the number of syndrome measurement repetitions, thereby making the EC gadget less costly.
\begin{itemize}
	\item {\it Redundant parity measurement.} We can detect faulty syndrome measurements by adding one more redundant parity measurement between the first and the last qubit [or between the first and the last row in case (iii)]. One round of measurements, therefore, involves $t+1$ parity measurements, giving a length-$(t+1)$ syndrome sequence of $\pm 1$. 
	
	In the absence of faults in the ancilla measurements, i.e., none of the measurement outcomes are flipped, the sequence always has an even number of $-1$s, regardless of the number of faults in the data qubits. An odd number of $-1$s arises only if there are measurement faults. Therefore, we call a syndrome sequence with an even number of $-1$s a \emph{valid} syndrome sequence, and one with an odd number of $-1$s an \emph{invalid} one. Obtaining an invalid sequence thus implies that at least one measurement fault has occurred during the current round of measurements. Also, an even number of measurement faults are needed to turn a valid syndrome sequence into another valid syndrome sequence. It means that every time we obtain a valid sequence, either it is the correct sequence without measurement faults, or there are at least two measurement faults in the sequence. Effectively, we can reduce the number of repetitions by half. 
	
	We then have the following situation for the three steps above. For step (i):  In the syndrome extraction, if the trivial sequence is obtained, we repeat the parity measurements up to $\lfloor \frac{t}{2} \rfloor + 1$ times before concluding that there are no damping errors. Whenever we obtain an invalid sequence, we conclude that there is at least one fault and continue as if we obtained a nontrivial sequence. For steps (ii) and (iii): Replacing $r_2$ by $r_2/2$ in the argument in the proof above, we arrive at $\tfrac{1}{8}(t+2)^2 + 1$ for the maximum number of repetitions. So, for the parity-check step and the $Z$-syndrome measurement, we repeat $\tfrac{1}{8}(t+2)^2 + 1$ times and choose the majority among the set of valid syndromes.

	\item {\it Keeping track of the number of faults.} Up till now, we have allowed for up to $t$ faults in the syndrome measurements. However, if we keep track of the number of faults that have already occurred in the \textsc{EC} gadget, then we can further reduce the number of repetitions needed before we violate the fault-tolerance pre-condition of $\leq t$ faults. For example, if there are at least $t_1$ faults already, then the syndrome measurements need to be fault-tolerant against only $t-t_1$ faults, thereby reducing the number of measurement repetitions. 
	
\end{itemize}

\section{Basic fault-tolerant components}\label{sec:ftgadgets}

We want to construct a universal set of logical gadgets that are fault tolerant. To do that, we first need some basic fault-tolerant operations, which we describe in this section. In the next section (App.~\ref{sec:universal}), we make use of these basic building-block operations to construct our logical gadgets. The basic operations we will construct here are the preparation and measurement gadgets.

\subsection{Measurement gadgets}\label{sec:logicalmeas}

Next, we describe fault-tolerant gadgets that implement the destructive measurements of logical $Z$ and $X$.

\subsubsection{Logical $Z$ measurement}
The logical $Z$ measurement can be realized by measuring each physical qubit in the $Z$ basis. In the absence of any errors, the $Z$ measurements on one row of qubits will give either all $-1$s or all $+1$s [see Eq.~\eqref{eq:01codewords}], which we refer to as a $(-1)$-row or a $(+1)$-row, respectively. The outcome of the logical-$Z$ measurement is $+1$ if there is an even number of $(+1)$-rows over the whole lattice [again, see Eq.~\eqref{eq:01codewords}]; it is $-1$ if there is an odd number of them. 

That this simple logical $Z$ measurement is fault tolerant can be seen as follows. We first note that we need only consider damping errors $\cF_a$s since $\cF_Z$ and $Z$ do not affect the $Z$ measurement.
A faulty physical $Z$ measurement damps a qubit to the $\ket{0}$ state and gives the outcome $+1$. A damped row is thus a mixture of $-1$s and $+1$s. For at most $t$ faults occurring during the measurements, a damped row has at most $t$ $-1$s and at least one $1$ since there are $t+1$ qubits in a row. A damped row is therefore always reliably detected.
Such a damped row should be regarded as $(+1)$-row when deducing the logical-$Z$ measurement outcome since a damping error has effect only when the qubit is originally in the state $|1\rangle$.

\subsubsection{Logical $X$ measurement}
The logical $X$ measurement can be similarly realized by measuring each physical qubit in the $X$ basis. However, unlike the logical $Z$ measurement, this transversal implementation is not fault tolerant. This can be understood by noticing that, in the no-fault situation, the logical $X$ outcome is equal to the parity of the measurement outcomes of any one row of qubits. A single damping error in a row randomizes the parity of that row. As we allow for damping in more than half the rows, we cannot deduce the correct logical outcome even by looking at the individual parity values for all rows.

Instead, we use the circuit in Fig.~\ref{fig:logicalXmeas} for a fault-tolerant logical $X$ measurement. For each row, an ancilla initialized in $\ket{+}$ is coupled to all qubits in the row via \textsc{CNOT} gates. Then, we disentangle one data qubit, say the first qubit, from the other qubits in the row. We then measure this qubit and the ancilla in the $X$ basis. The other qubits in the row are measured in the $Z$ basis.
\begin{figure}[t]
	\centering
	\begin{quantikz}[row sep=0.05cm, column sep = 0.3cm]
    		\lstick{$\ket{+}$} & \ctrl{1} & \ctrl{2} & \qw \hspace{1mm} \dots \hspace{1mm} & \ctrl{4} &  \meterD{X} \\
    		\lstick[wires=4]{$\text{Data qubits}$} & \targ{} & \qw & \qw \hspace{1mm} \dots \hspace{1mm} & \qw &  \ctrl{1} & \qw \hspace{1mm} \dots \hspace{1mm} & \ctrl{3} & \meterD{X} \\
    		\lstick{} & \qw & \targ{} & \qw \hspace{1mm} \dots \hspace{1mm} & \qw & \targ{} & \qw \hspace{1mm} \dots \hspace{1mm} & \qw & \meterD{Z} \\
    		\vdots \\
    		\lstick{} & \qw & \qw & \qw \hspace{1mm} \dots \hspace{1mm} & \targ{} & \qw & \qw \hspace{1mm} \dots \hspace{1mm} & \targ{} & \meterD{Z} \\
	\end{quantikz}
	\caption{Logical $X$ measurement on a single row. The logical outcome is given by the outcomes of the two $X$ measurements. The row is \emph{valid} only if both the $X$ measurements give the same outcome and every $Z$ measurement gives $+1$.}
	\label{fig:logicalXmeas}
\end{figure}
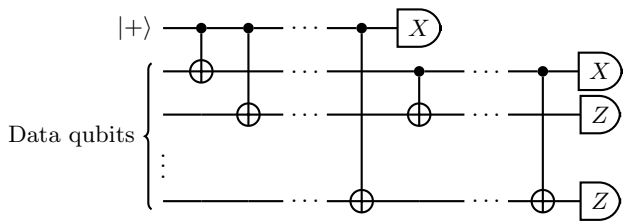 

For input states $\ket{\pm}_L$, the only states of relevance to the logical $X$ measurement, the joint states of the row with the ancilla (denoted $a$) before the final $X$ and $Z$ measurements are, in the absence of any faults, $\ket{\pm}_{a}\ket{\pm}\ket{0}^{\otimes t}$. The outcome of the ancilla or the first qubit is, therefore, the logical outcome. If there are faults, the measurement outcomes can deviate from the ideal case. We say that a row is invalid if at least one of the $Z$ measurements has outcome $-1$, or if the outcome of the ancilla measurement does not match the outcome of the first data qubit. The logical outcome is then determined as the majority outcome among the valid rows. 

The fault tolerance of the measurement circuit in Fig.~\ref{fig:logicalXmeas} can be understood as follows. For each row, a single damping error or a single $\cF_Z$ anywhere either makes the row invalid or does not affect the outcome of the measurements performed on the row, as explained below.
\begin{itemize}
	\item A fault in the initialization or measurement of the ancilla may randomize the outcome of the ancilla measurement, but it has no effect on the data qubits. Therefore, if the outcome of the ancilla measurement matches that of the first data qubit, the latter will be correct. 
	\item A fault in the measurement of the first data qubit also randomizes its outcome, but it has no effect on the ancilla's outcome. Therefore, the row's outcome is correct if the two outcomes match. 
	\item A fault in the \textsc{CNOT}s or an error in the incoming state results in at least one of the $Z$ measurements having outcome $1$. 
\end{itemize}
A second-order error, which can either be a $Z$ error, or two $\cF_{z}$ errors, or two $\cF_a$ errors, might lead to a valid row with an incorrect outcome. For $s$ errors in the input and $r$ faults in the measurement gadget such that $s+r\leq t$, more than half the rows will still be valid and correct, and hence, the logical outcome is always determined correctly from majority voting.

\subsection{Preparation gadgets} \label{sec:prep}
Finally, we describe fault-tolerant preparation of the logical state $\ket{+}_L$.

\subsubsection{Preparation of cat state}
The $(t+1)$-qubit cat state is defined as
\begin{equation}
	\ket{\beta_0} \equiv \tfrac{1}{\sqrt{2}}{\left(\ket{0}^{\otimes (t+1)} + \ket{1}^{\otimes (t+1)}\right)}.
\end{equation}
This state serves as the cornerstone for the preparation of other logical states, for example, the logical $\ket{+}_L$ and the logical resource states (see App.~\ref{sec:universal}). To prepare this state, the idea is to measure the operator $X_1X_2\dots X_{t+1}$ on the input state $\ket{0}^{\otimes (t+1)}$ and post-select the outputs based on the measurement outcomes. We measure $X_1X_2\dots X_{t+1}$ using the circuit in Fig.~\ref{fig:catprep}, where an ancillary qubit initialized in the $\ket{+}$ state is coupled to the data qubits through \textsc{CNOT}s and finally measured in the $X$ basis. We do the measurement twice (see below for the reason), i.e., the outgoing state after the first measurement becomes the input for the second measurement. Finally, we pass the outgoing state after the second measurement through the damping syndrome extraction unit, as described in Sec.~\ref{sec:subcircuit} for the \textsc{EC} gadget. The final output state is accepted only if all the measurements, namely, two $X$-type measurements and all the parity measurements, give $+1$ outcomes. 

\begin{figure}[t]
	\includegraphics[scale=0.38]{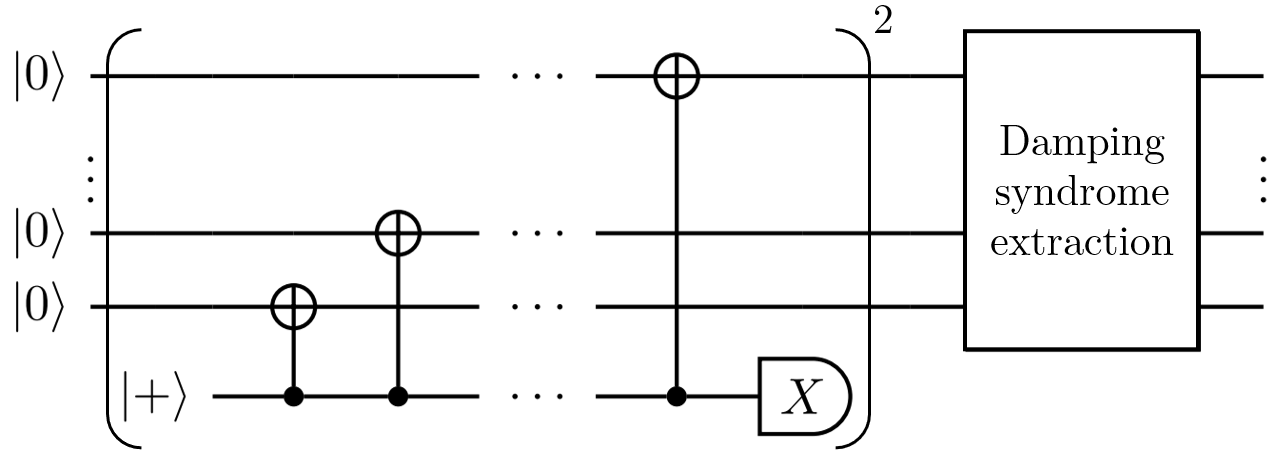}
	\caption{Preparation and verification of the cat state. The $X_1X_2\dots X_{t+1}$ measurement is done twice followed by the damping syndrome extraction unit. The output state is accepted if all the measurements give $+1$ outcomes.}
	\label{fig:catprep}
\end{figure}

That this preparation procedure satisfies the fault-tolerance condition (P2) of App.~\ref{app:ftprops} can be understood as follows. We consider up to $t$ faults anywhere in the preparation gadget, which includes two $X_1X_2\dots X_{t+1}$ measurements and the damping syndrome extraction unit. We will check that the gadget either accepts the output state with the correct fault-tolerance property or rejects it. 
\begin{itemize}
	\item $Z$ or $\cF_z$ errors pass through all the checks but they do not spread. Hence, $t$ $Z$ or $\cF_Z$ errors remain as $t$ $Z$ or $\cF_Z$ errors in the output state. The errors may change the outcome of the $X$ measurements, but they have no effect on the output state. 
	\item Damping errors $\cF_a$ on the data qubits during the two measurements cause $\cF_a$ and/or $X$ errors in the outgoing state of the second measurement. Damping of the ancillary qubits in the middle of the measurements also propagates $X$ errors to the data qubits. The damping syndrome extraction unit (even with faults inside it) ensures that these errors will be detected and the output state will be rejected. Moreover, $s$ faults inside the damping syndrome extraction unit lead to at most $s$ damping errors in the output state (see Sec.~\ref{sec:subcircuit}). Thus, if there are only faults inside the syndrome extraction, the output state is accepted with no more errors than the number of faults inside the unit.
	\item More worrying are damping errors after the preparation of the first ancillary qubit, or right before the $X$ measurement of the first ancillary qubit. The outgoing state of the first $X_1X_2\dots X_{t+1}$ measurement is $\ket{0}^{\otimes (t+1)}$ for the first case and $\ket{1}^{\otimes (t+1)}$ for the second. Note that both states pass the damping syndrome extraction unit. Moreover, $\ket{0/1}^{\otimes (t+1)}$ can be thought of as the cat state $\ket{\beta_0}$ with $(I\pm Z)$ error. Therefore, if we stop after the first measurement, there might be a chance that one damping error causes a $Z$ error in the output state and there is no way to tell if this error comes from a first-order damping error. This violates the fault-tolerance property (P2) stated in App.~\ref{app:ftprops}. The second $X_1X_2\dots X_{t+1}$ measurement is added to remedy this. It projects the $\ket{0/1}^{\otimes (t+1)}$ state to the cat state (half of the time) and hence removes the $Z$ error. The $\ket{0/1}^{\otimes (t+1)}$ state is still possible as an output state if both of the ancillary qubits are damped. However, this is a second-order event and the fault-tolerance condition is still fulfilled.
\end{itemize}

\subsubsection{Preparation of $\ket{+}_L$}
Observe that $\ket{+}_L = \ket{\beta_0}^{\otimes (t+1)}$. Thus, the preparation of $\ket{+}_L$ can be implemented by fault-tolerantly preparing $t+1$ copies of the cat state.

\section{A universal set of logical gadgets}\label{sec:universal}

We are now ready to construct a universal set of logical gadgets, using the basic fault-tolerant components of the previous section. In particular, our universal set of gadgets comprises,
\begin{enumerate}
	\item preparation of the $|0\rangle_L$ and $|+\rangle_L$ states;
	\item measurement of logical $X$ and $Z$;
	\item two-qubit logical gate \textsc{CZ};
	\item single-qubit logical gates $H$, $S$, and $T$.
\end{enumerate}
Here, the logical $H$ is the Hadamard gate, $\overline H\equiv |+\rangle_L\langle 0|+|-\rangle_L\langle 1|$; the logical $S$ is the phase gate, $\overline S\equiv |0\rangle_L\langle 0|+\mi |1\rangle_L\langle 1|$; and the logical $T$ (or $\pi/8$) is the gate $\overline T\equiv |0\rangle_L\langle 0| + \mathrm{e}^{\mi\pi/4}|1\rangle_L\langle 1|$.

We note that our construction requires no logical $X$ or $Z$ gates. $\overline X$ and $\overline Z$ can be tracked virtually when the operations are all Clifford, with no need for physical gates. In the presence of the single non-Clifford logical $T$ gate, we note that $\overline T\overline Z\overline T^\dagger=\overline Z$, while $\overline T\overline X\overline T^\dagger=\overline S\overline X$. $\overline Z$ thus simply commutes past $\overline T$, while for the $\overline X$, we apply the $\overline S$ and commute the remaining $\overline X$ past the operations.

Apart from the preparation of $\ket{0}_L$ and the single-logical-qubit gadgets, all the other gadgets are already constructed in the earlier sections of the paper. We focus our attention on the single-logical-qubit gates first. Our construction of all three single-logical-qubit gadgets makes use of the same teleportation structure shown in Fig.~\ref{fig:HST}. Each requires a resource state as input: $|+\rangle_L$ for $\overline H$, $|\Phi_S\rangle\equiv\frac{1}{\sqrt{2}}{\left(\ket{0}_L + \mi\ket{1}_L\right)}$ for $\overline S$, and $|\Phi_T\rangle \equiv \frac{1}{\sqrt{2}}{\left(\ket{0}_L + \mathrm{e}^{\mi\pi/4}\ket{1}_L\right)}$ for $\overline T$. The circuits also require the use of the logical $X$ measurement gadget, as well as an additional logical gadget, conditioned on the measurement outcome, to complete the gate teleportation. 

\begin{figure}[t]
	\centering
	\begin{quantikz}
    		\lstick{$\ket{\Psi}$} & \gate{\overline{g}_1} & \ctrl{1} & \meterD{\overline{X}} \vcw{1} \\
    		\lstick{$\ket{\Phi}$} & \qw & \control{} &  \gate{\overline{g}_2} & \qw \text{$\overline{G}\ket{\Psi}$}
	\end{quantikz}
	\caption{Single-logical-qubit gadget implementing the gate $\overline G=\overline H, \overline S,$ or $\overline T$. $|\Psi\rangle$ is the incoming state on which $\overline G$ is to be applied. For $\overline G=\overline H$, $|\Phi\rangle=|+\rangle_L$, $\overline g_1$ is the identity, and $\overline g_2=\overline X$; for $\overline G=\overline S$, $|\Phi\rangle=|\Phi_S\rangle$, $\overline g_1=\overline H$, and $\overline g_2=\overline Y=\overline{XZ}$; for $\overline G=\overline T$, $|\Phi\rangle=|\Phi_T\rangle$, $\overline g_1=\overline H$, and $\overline g_2=\overline{SX}$. The gate $\overline g_2$ is applied, conditioned on the outcome of the $\overline X$ measurement. We note that the $\overline X$ and $\overline Z$ contained in the $\overline g_2$s do not need to be applied physically; see the main text.
	}
	\label{fig:HST}
\end{figure}
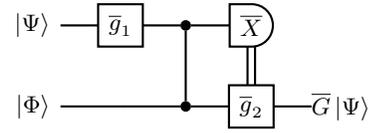

Note that we already have fault-tolerant gadgets for every operation in the gate-teleportation circuits in Fig.~\ref{fig:HST}, as well as for the preparation of $\ket{+}_L$. We thus already have a fault-tolerant $\overline{H}$. This also gives us a fault-tolerant preparation of $\ket{0}_L$: We apply the fault-tolerant $\overline{H}$ to the fault-tolerantly prepared $\ket{+}_L$. What remains then is to describe the fault-tolerant preparation of $\ket{\Phi_S}$ and $\ket{\Phi_T}$. The fault tolerance of the teleportation gadgets then simply follows from the fact that their components are all separately fault tolerant. 

The preparation of $\ket{\Phi_S}$ and $\ket{\Phi_T}$, written together as $\ket{\Phi_{S/T}}$, starts by first preparing the single-row states $\ket{\beta_{S/T}}$, defined as
\begin{align}
	\ket{\beta_S} &= \tfrac{1}{\sqrt{2}}(\ket{0}^{\otimes (t+1)} + \mi \ket{1}^{\otimes (t+1)}), \nonumber\\
	\textrm{and }\quad	\ket{\beta_T} &= \tfrac{1}{\sqrt{2}}(\ket{0}^{\otimes (t+1)} + e^{\mi \pi/4} \ket{1}^{\otimes (t+1)}).
\end{align}
The idea is to prepare a cat state first and apply a single-qubit gate to get the correct phase factor. We use the circuit shown in Fig.~\ref{fig:STprep}, similar to the one used to prepare the cat state. 
We start by doing the measurement of $X_1X_2\dots X_{t+1}$ a total of $t+1$ times, by feeding the outgoing state of each measurement as the input of the next measurement. The outgoing state of the last measurement is passed through the damping syndrome extraction unit and finally, we apply an $S$ or a $T$ gate (depending on whether we are preparing $\ket{\beta_S}$ or $\ket{\beta_T}$) to one of the qubits. We accept the output state to be properly prepared only when all the measurement outcomes are $+1$. 

The fault tolerance of this preparation step can be understood in a very similar way to the analysis given for the cat state. The only difference is the number of repetitions for $X_1X_2\dots X_{t+1}$ measurement. We noted previously that a damping error on the ancilla qubit may cause the outgoing state of a measurement to be one of the trivial states $\ket{0/1}^{\otimes (t+1)}$. Applying a $S/T$ gate to this state is meaningless since the relative phase is already lost. Repeating the measurement $t$ times ensures that at least one of them has its ancilla qubit undamped. The accepted output (if no damping after the damping extraction) is thus a genuine cat state and the single-qubit gate $S/T$ should add the correct relative phase to the state. 

\begin{figure}[t]
	\centering
	\includegraphics[scale=0.33]{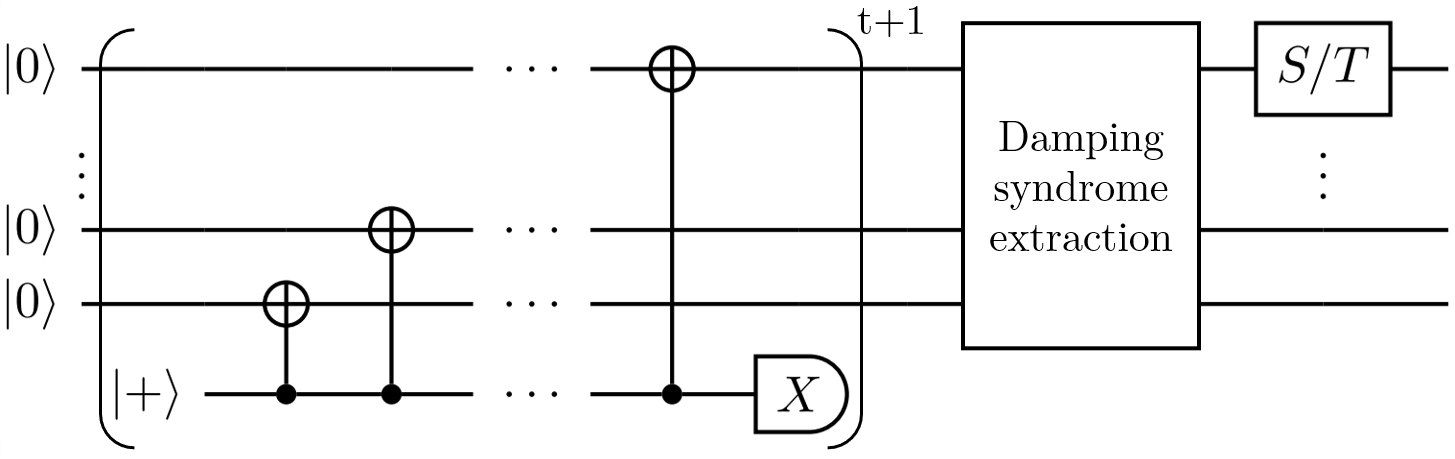}
	\caption{Preparation of single-row states $\ket{\beta_{S/T}}$. The $X_1X_2\dots X_{t+1}$ measurement is done $t+1$ times before applying a damping syndrome extraction and a single-qubit gate $S/T$. The output state is accepted if all the measurements have $+1$ outcomes.}
	\label{fig:STprep}
\end{figure} 

Now, we need to transform $\ket{\beta_{S/T}}$ into the logical states $\ket{\Phi_{S/T}}$ on the full lattice. To do this, we first prepare the other $t$ rows in a state denoted as $\ket{0}_L^{t\times(t+1)}$, that is, a logical $\ket{0}_L$ state of a $t\times(t+1)$ Bacon-Shor code. This can be done as follows. We prepare $\ket{0}_L$ fault-tolerantly as described earlier and measure one row in the $Z$ basis. If the outcomes are all $-1$'s, we obtain the desired state $\ket{0}_L^{t\times(t+1)}$; otherwise, if the outcomes are all $1$s or a mixture of $0$s and $1$s, we have, instead, the state $\ket{1}_L^{t\times(t+1)}$. In this case, we can simply keep track of the Pauli error in the software. 
Next, we measure the $\overline{XX}$-stabilizer between the row prepared in the $\ket{\beta_{S/T}}$ state and any row of the $\ket{0}_L^{t\times(t+1)}$ state. This should be done using a subcircuit discussed earlier in Sec.~\ref{sec:ecunit}. 

An accepted $\ket{\beta_{S/T}}$ state may still have damping errors caused by faults inside the damping extraction or the $S/T$ gate at the end of the circuit in Fig.~\ref{fig:STprep}. In this case, the relative phase in the $\ket{\beta_{S/T}}$ state would be annihilated before it is transferred to the full lattice. The damping syndrome extraction in the subcircuit is therefore needed to detect this error. We thus accept the output state from the subcircuit only if its damping syndrome extraction unit detects no damping.

\section{The \textsc{CCZ} gadget}\label{sec:ccz}

The previous section gives a universal set of logical operations sufficient for arbitrary computation with the Bacon-Shor code in the presence of amplitude-damping noise. Nevertheless, we point out that one can expand the set of fault-tolerant operations to include a three-qubit logical gadget---the \textsc{CCZ} gate---on a general $(t+1)\times (t+1)$ Bacon-Shor code that it is tolerant to $\leq t$ errors arising from the amplitude-damping noise. The \textsc{CCZ} gate also provides an alternate pathway to universality, since the \textsc{CCZ} together with the Hadamard gate implemented by teleportation forms a universal set of logical gates, without requiring the $S$ and $T$ gadgets described earlier.  

Logical \textsc{CCZ} is implemented by applying a set of physical \textsc{CCZ}s transversally for $t+1$ times. In each step the rows of a control lattice are connected to the columns of the other control lattice and, in turn, the rows of the second control lattice are connected to the rows of the target lattice. In each step, the connections to the rows of the target lattice are permuted in a cyclic manner. 

Fig.~\ref{fig:CCZ} illustrates an example of the logical \textsc{CCZ} for the $3\times 3$ Bacon-Shor code. Here, we intersperse the \textsc{CCZ} gadget with EC units between two executions of the transversal \textsc{CCZ}s. The syndrome bits allow us to decide how to correct the final qubits for a fault-tolerant \textsc{CCZ}. For instance, note that a damping error on a qubit in the first lattice could propagate a \textsc{CZ} operation on the qubits in the other two blocks,
\begin{eqnarray}
	\textsc{CCZ} (\cF_a \otimes I \otimes I) \textsc{CCZ}&= &\cF_a\otimes \textsc{CZ}.
\end{eqnarray}
In such a case, we use the syndrome bit information to decide how to apply \textsc{CZ} to undo the propagated errors. Other such checks can be done to ensure the fault tolerance of the gadget, but we leave the details for the reader to verify.

\begin{figure}[t]
	\includegraphics[trim=10mm 70mm 25mm 10mm, clip, scale=0.3]{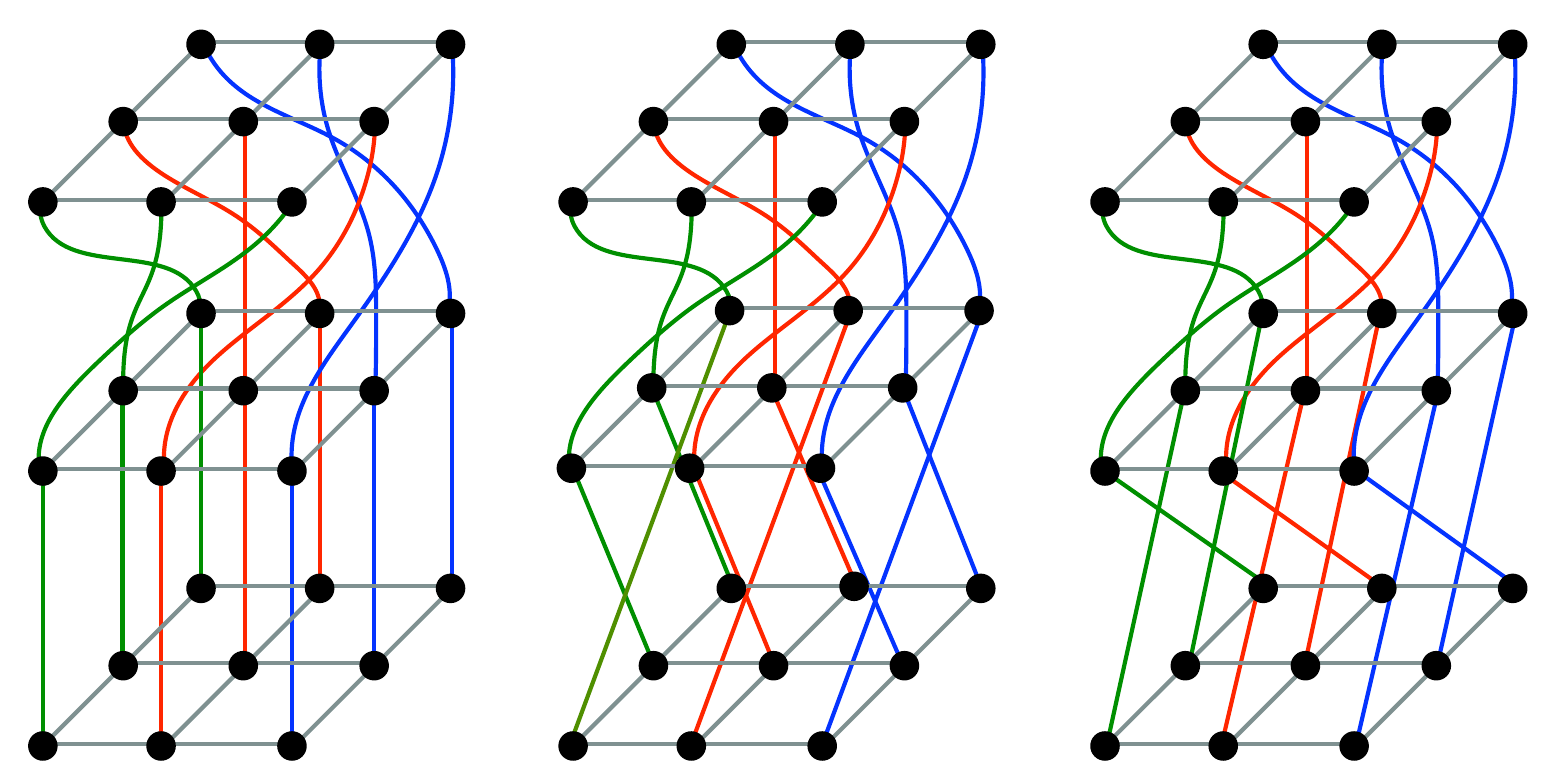}
	\caption{Implementation of the encoded \textsc{CCZ} gate for the $3 \times 3$ [$(t+1)\times (t+1)$] Bacon-Shor code in three ($t+1$) steps, going from left to right. The qubits in different code blocks that participate in the same physical \textsc{CCZ} gate are linked by the same color.} 
	\label{fig:CCZ}
\end{figure}

\section{Pseudothreshold calculation}\label{sec:ft_pth}
Finally, we prove the result stated in the main text, that the EC gadget has $O(t^6)$ fault locations; this was used in Sec.~\ref{sec:threshold} to estimate the pseudothreshold.
We first count the number of locations in a single \textsc{EC} subcircuit [see Fig.~\ref{fig:ECBigPicture}(b)], assuming the efficient syndrome extraction strategy described in App.~\ref{app:ft_syndrome}:
\begin{itemize}
	\item Coupling: There are $2t+1$ locations per row, and hence $2(2t+1)$ locations in total.
	\item Damping syndrome extraction: For a single row, each round of parity measurements is done in $3$ time steps, the first two steps involving the coupling of $t+1$ data qubits to $t+1$ ancilla, and the last step is the measurement of the ancilla. During that time, the additional ancillary qubits are all idle. Therefore, there are $3(3t+2)$ locations per round. At most $\lfloor t/2 \rfloor+1$ rounds are performed, and hence, there are $6(3t+2)(\lfloor t/2 \rfloor+1)$ possible locations for both rows in a given subcircuit. 
	\item Recovery for $X$ errors: We can delay this step until the end of the subcircuit.
	\item $\overline{X}\overline{X}$ stabilizer measurement: There are $2(2t+1)+2$ qubits involved (including extended data, ancillary, and flag qubits) and the measurement is done in $2(2t+1)+2$ time steps including ancilla-flag decoupling and measurements at the end. Together with the 4 locations in the preparation of ancillary and flag qubits, there are $16(t+1)^2+4$ locations in total.
	\item Parity checks: In this step, parity measurements are done on the full set of extended data qubits. Hence, each round has $3\times 2(2t+1)$ locations. If we repeat the parity measurements until a valid syndrome sequence appears $(\lfloor t/2 \rfloor+1)$ times, in the worst case we need $(t+1)(\lfloor t/2 \rfloor+1) + 1$ rounds. A more optimized scheme needs only $\lfloor (t+2)^2/8 \rfloor+ 1 $ rounds. In total, there are at most $12(2t+1)(\lfloor (t+2)^2/8 \rfloor+ 1)$ locations for both rows. 
	\item Decoupling and recovery: We decouple and measure all the extra qubits in the $Z$ basis. The recovery for $X$ errors takes up one more time step. In total, there are $6(2t+1)$ locations. 
\end{itemize}
Taking all the above steps together, in total there are at most $N_{\text{sub}} = 16t^2+72t+40 + 12(2t+1)\lfloor\frac{(t+2)^2}{8}\rfloor + 6(3t+2)(\lfloor t/2 \rfloor+1)$ locations in a subcircuit. 

For odd $t$, there are an even number of rows in the lattice and we can perform $(t+1)/2$ subcircuits in parallel. Recall from Sec.~\ref{sec:ecunit} that one round of $Z$-syndrome extraction needs $t+1$ subcircuits and this is repeated until a valid sequence appears a certain number of times. Since a second-order event is required to convert a valid $Z$-syndrome sequence into another valid sequence, in the worst case, similarly to the parity-measurement step in a subcircuit, we need to repeat it $\lfloor (t/2+2)^2/4 \rfloor+ 1$ times. Therefore, in total, there are at most $N_{\text{odd}} = (\lfloor(t/2+2)^2/4\rfloor + 1)(t+1)N_{\text{sub}}$ locations in an \textsc{EC} gadget for odd $t$. 

For even $t$, there are an odd number of rows in the lattice. Hence, we can perform $t/2$ subcircuits in parallel with one row left idle. The total number of locations in one round is therefore $\tfrac{t}{2}N_{\text{sub}}+N_{\text{idle}}$, where $N_{\text{idle}} = (t+1)\left(3(t+1)\lfloor t/2 \rfloor+4t+9\right)$ is the number of locations in the idling row. As in the odd $t$ case, we also need to repeat the syndrome extraction $\lfloor (t/2+2)^2/4 \rfloor+ 1$ times in the worst case. However, in this case we can schedule the subcircuits so that after $t+1$ parallel steps, we measure the $Z$-syndrome $t/2$ times. Therefore, in total, there are roughly at most $N_{\text{even}} = \lfloor t/8+2\rfloor(t+1)(\tfrac{t}{2}N_{\text{sub}}+N_{\text{idle}})$ locations in an \textsc{EC} gadget for even $t$.

\end{document}